\documentclass[11pt,a4paper]{article}
\usepackage[utf8]{inputenc}
\usepackage[T1]{fontenc}
\usepackage{amsfonts}
\usepackage[vlined,ruled,linesnumbered]{algorithm2e}
\usepackage[pdftex]{graphicx}
\usepackage{amsthm}
\usepackage{amsmath}
\usepackage{amssymb}
\usepackage{url}
\usepackage{cite}
\usepackage{booktabs}
\usepackage{hyperref}
\usepackage{mathrsfs}
\usepackage[vmargin=3cm, hmargin=3cm]{geometry}
\usepackage{paralist}

\usepackage{microtype}
\DeclareMathOperator{\rand}{rand}

\DeclareMathOperator{\E}{\mathbb{E}}
\DeclareMathOperator{\Pb}{Pr}

\DeclareMathOperator{\oc}{oc}
\DeclareMathOperator{\ms}{ms}

\DeclareMathOperator{\miss}{MISS}

\DeclareMathOperator{\ev}{EV}

\DeclareMathOperator{\lru}{LRU}
\DeclareMathOperator{\mru}{MRU}
\DeclareMathOperator{\opt}{OPT}
\DeclareMathOperator{\T}{T}

\DeclareMathOperator{\lpr}{LPR}
\DeclareMathOperator{\kl}{K-L}

\DeclareMathOperator{\mf}{\enspace .}
\DeclareMathOperator{\mc}{\enspace ,}

\newcommand{\deq}{\mathrel{\mathop:}=}

\newcommand{\Z}{\mathbb{Z}}

\newcommand{\Pscal}{{\sc Scal}}
\newcommand{\tone}{}
\newcommand{\ttwo}{'}

\newcommand{\reals}{\mathbb{R}}
\newcommand{\comment}[1]{}

\newtheorem{theorem}{Theorem}

\newtheorem{corollary}{Corollary}
\newtheorem{proposition}{Proposition}
\newtheorem{lemma}{Lemma}

\theoremstyle{definition}
\newtheorem{definition}{Definition}
\newtheorem{remark}{Remark}

\theoremstyle{remark}

\def\vec#1{\mathchoice{\mbox{\boldmath$\displaystyle#1$}}
  {\mbox{\boldmath$\textstyle#1$}}
  {\mbox{\boldmath$\scriptstyle#1$}}
  {\mbox{\boldmath$\scriptscriptstyle#1$}}}

\newcounter{romen} \newenvironment{romenum}{\setcounter{romen}{1}\def\item{
    (\roman{romen})\ \stepcounter{romen}}}{\newline}

\title{Optimal Eviction Policies for Stochastic Address
  Traces\footnote{Published in Theoretical Computer Science, 2013:
    \url{http://dx.doi.org/10.1016/j.tcs.2013.01.016}}}

\author{Gianfranco Bilardi\thanks{\texttt{bilardi@dei.unipd.it}} \and Francesco
  Versaci\thanks{\texttt{versaci@par.tuwien.ac.at}. {Currently at Vienna
      University of Technology. Supported by PAT-INFN Project
      \emph{AuroraScience}, by MIUR-PRIN Project \emph{AlgoDEEP}, and by the
      University of Padova Projects \emph{STPD08JA32} and \emph{CPDA099949}.}}}
\date{University of Padova}

\begin{document}
\maketitle

\begin{abstract}
  The eviction problem for memory hierarchies is studied for the
  Hidden Markov Reference Model (HMRM) of the memory trace,
  showing how miss minimization can be naturally formulated in
  the optimal control setting.
    In addition to the traditional version assuming a buffer of fixed
  capacity, a relaxed version is also considered, in which buffer
  occupancy can vary and its average is constrained.
    Resorting to multiobjective optimization, viewing occupancy as a
  cost rather than as a constraint, the optimal eviction policy is
  obtained by composing solutions for the individual addressable
  items.

  This approach is then specialized to the Least Recently Used Stack Model
  (LRUSM), a type of HMRM often considered for traces, which includes $V-1$
  parameters, where $V$ is the size of the virtual space. A gain optimal policy
  for any target average occupancy is obtained which (i) is computable in time
  $O(V)$ from the model parameters, (ii) is optimal also for the fixed capacity
  case, and (iii) is characterized in terms of priorities, with the name of
  Least Profit Rate (LPR) policy. An $O(\log C)$ upper bound (being $C$ the
  buffer capacity) is derived for the ratio between the expected miss rate of
  LPR and that of OPT, the optimal off-line policy; the upper bound is tightened
  to $O(1)$, under reasonable constraints on the LRUSM parameters. Using the
  stack-distance framework, an algorithm is developed to compute the number of
  misses incurred by LPR on a given input trace, simultaneously for all buffer
  capacities, in time $O(\log V)$ per access.

  Finally, some results are provided for miss minimization over a
  finite horizon and over an infinite horizon under bias optimality,
  a criterion more stringent than gain optimality.

\end{abstract}

\vspace{5mm}
\noindent{\bf Keywords}: Eviction policies, Paging, Online problems,
Algorithms and data structures, Markov chains, Optimal control,
Multiobjective optimization.


\section{Introduction to Eviction Policies for the Memory Hierarchy}
\label{chap:an-optimal-control}

The storage of most computer systems is organized as a hierarchy of
levels (currently, half a dozen), for technological and economical
reasons \cite{HennessyP06} as well as due to fundamental physical
constraints \cite{BilardiP95}. Memory hierarchies have been
investigated extensively, in terms of hardware organization
\cite{Fotheringham61,Przybylski90,HennessyP06}, operating systems
\cite{SilberschatzGG05}, compiler optimization
\cite{AllenK02,Wolfe95,GuoGP03}, models of computation
\cite{Savage97}, and algorithm design \cite{AggarwalACS87}.  A central
issue is the decision of which data to keep in which level.  It is
customary to focus on two levels (see Fig.~\ref{two-lev-hier}),
respectively called here the \emph{buffer} and the \emph{backing
  storage}, the extension to multiple levels being generally
straightforward. (We adopt the neutral names used in the seminal paper
of Mattson et al.\ \cite{MattsonGST70}, since the concepts introduced
have application at each level of the memory hierarchy: register
allocation, CPU caching, memory paging, web caching, etc.) We assume
that an engine generates a sequence of access requests $a_1, a_2
\ldots, a_t, \ldots$ for \emph{items} (equally sized blocks of data)
stored in the hierarchy.  A request is called a \emph{hit} if the
requested item is in the buffer and a \emph{miss} otherwise. Upon a
miss, the item must be brought into the buffer, a costly operation.
If the buffer is full, the requested item will replace another
item. We are interested in an \emph{eviction policy} that selects the
items to be replaced so as to minimize subsequent misses.

The \emph{MIN policy} of Belady \cite{Belady66} and the \emph{OPT
  policy}\footnotemark\ of Mattson, Gecsei, Sluts, and Traiger
\cite{MattsonGST70}, minimize the number of misses in an off-line
setting. \footnotetext{OPT: Upon a miss, if the buffer contains items that will
  not be accessed in the future then evict (any) one of them, else evict the
  (unique) item whose next access is furthest in the future.} Since, in
practical situations, the address trace unfolds with the computation, eviction
decisions must be made \emph{on-line}.  Dozens of on-line policies have been
proposed and implemented in hardware or within operating systems.
Somewhat schematically, we can say that implemented policies have evolved mostly
experimentally, by benchmarking plausible proposals against relevant workloads
\cite{Przybylski90}.  Most of these policies are variants of the \emph{Least
  Recently Used} (LRU) policy. Departure from pure LRU is motivated to a large
extent by its high implementation cost.  However, several authors have also
explored variants of LRU that incur fewer misses, at least on some
workloads~\cite{MegiddoM04}.
\begin{figure}
  \centering  \includegraphics[width=2.3cm]{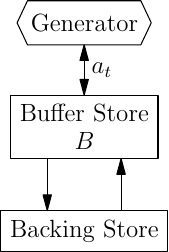}
  \caption{Two-level memory hierarchy.}
  \label{two-lev-hier}
\end{figure}

Theoretical investigations have focused mostly on two objectives: (a) to
``explain'' the practical success of LRU-like policies and (b) to explore the
existence of better policies.  One major question is how to model the input
traces (i.e., the lists of requested memory references). An interesting
perspective, proposed by Koutsoupias and Papadimitriou \cite{KoutsoupiasP00}, is
to model traces by a class of stochastic processes, all to be dealt with by the
same policy.  For a given stochastic model, two metrics help assess the quality
of a policy: (i) the expected number of misses and (ii) its ratio with the
expected number of misses incurred by OPT, called the \emph{competitive ratio}.
The competitive ratio of a policy with respect to a class of stochastic
processes is defined in a worst-case sense, maximizing over the class.

Competitive analysis was proposed by Sleator and Tarjan \cite{SleatorT85} for
the class of all possible traces (all stochastic processes), for which they
showed that the competitive ratio of any on-line policy is at least the buffer
capacity $C$, a value actually achieved by, e.g., LRU and FIFO. While
theoretically interesting, this results shed little light on what is
experimentally known. For example, if we restrict the possible traces to the
ones which actually emerge in practical applications, the miss ratio between LRU
and OPT is much smaller than $C$, being typically around 2 and seldom exceeding
4.
Borodin et al.\ \cite{BorodinIRS95} restrict the class of possible
traces to be consistent with an underlying ``access graph'' that
models the program locality (nodes correspond to items and legal
traces correspond to walks); a heuristic policy based on this model has
been developed by Fiat and Rosen \cite{FiatR97} and has been shown to
perform better than LRU, relative to some benchmarks.

Albers et al.\ \cite{AlbersFG05} propose a trace restriction based on
Denning's working set concept \cite{Denning68}: for a given function
$f(n)$, the (average or maximum) number of distinct items referenced
in $n$ consecutive steps is at most $f(n)$. LRU is proved to be
optimal in both the average and maximum cases.

Koutsoupias and Papadimitriou \cite{KoutsoupiasP00} have considered
the class $\Delta_{\epsilon}$ ($0 < \epsilon \leq 1$) of stochastic
processes such that, given any prefix of the trace, the probability to
be accessed next is at most $\epsilon$ for any item.  A
combinatorially rich development establishes that LRU achieves minimum
competitive ratio, for each value of $\epsilon$.  A careful analysis
of how the competitive ratio depends upon both $\epsilon$ and $C$ has
been provided by Young \cite{Young98}: in particular, the competitive
ratio increases with $\epsilon$. Qualitatively speaking, buffering is
more efficient exactly when the items in the buffer are more likely to
be referenced than those outside the buffer. Hence, for traces where
LRU (or any policy) exhibits good performance (that is, few misses),
$\epsilon$ must be correspondingly high. Then, the $\Delta_{\epsilon}$
competitive ratio is also high, unlike what observed in practice.  For
a more quantitative appraisal of the issue, consider that, at any level
of the memory hierarchy of real systems, the average miss ratio is
typically below 1/4 (even well below 1\% in main memory). Let then
$C_{1/2}$ be the buffer capacity at which the miss rate would be
1/2. It is easy to see that it must be $\epsilon C_{1/2} \geq 1/2$.
The actual buffer capacity $C$ is typically considerably larger than
$C_{1/2}$, say $C \geq 4 C_{1/2}$, by the rule of thumb that
quadrupling the cache capacity halves the miss ratio
\cite{Przybylski90}.  Then, $\epsilon C \geq 2$. By the bounds of
Young \cite{Young98}, the $\Delta_{\epsilon}$ competitive ratio of LRU
is at least $C/2$, which is not much more informative than the value
$C$ of \cite{SleatorT85}.

To ``explain'' both the low miss rate and the low competitive ratio of
LRU in practical cases, the approach of \cite{KoutsoupiasP00} requires
some restriction to the class of stochastic processes, in order to
capture temporal locality of the reference trace, while essentially
assigning low the probability to those traces for which OPT vastly
outperforms LRU.  The model and results of Becchetti
\cite{Becchetti04} can be viewed as an exploration of this direction.

Following what could be viewed as an extreme case of the approach
outlined above, a number of studies have focused on specific
stochastic processes (the class is a singleton), with the objective of
developing (individually) optimal policies for such processes.  In
\cite{FranaszekW74}, the address trace $a_1,a_2, \ldots, a_t, \ldots$
is taken to be a sequence of mutually independent random variables, a
scenario known as the Independent Reference Model (IRM).  While
attractive for its simplicity, the IRM does not capture the
cornerstone property that memory hierarchies relay upon: the temporal
locality of references.  To avoid this drawback, the LRU-Stack Model
(LRUSM) of the trace \cite{OdenS72,SpirnD72} has been widely
considered in the literature
\cite{TurnerS77,EffelsbergH84,KobayashiM89} and is the focus of much
attention in this paper.  Here, the trace is statistically
characterized by the probability $s(j)$ of accessing the $j$-th most
recently referenced item.  The values of $j$ for different accesses
are assumed to be statistically independent. Throughout, we assume
that the items being referenced belong to a finite virtual space of
size $V$.  If $s$ is monotonically decreasing, LRU is optimal
\cite{SpirnD72} (for an infinite trace); in general, the optimal
policy varies with $s$ and can differ from LRU, as shown by Woof,
Fernandez and Lang\cite{WoodFL77,WoodFL83}. Smaragdakis, Kaplan, and
Wilson\cite{SmaragdakisKW99,SmaragdakisKW03} introduced the EELRU
policy as a heuristic inspired by the LRUSM. The model of Becchetti
\cite{Becchetti04} mentioned above is similar to LRUSM,  essentially
dropping the assumption of independence, while restricting the form of
the conditional distribution of accessing the $j$-th most recent item,
given the value of the past trace. LRU is compared to OPT under this
model.

A different generalization of the IRM model is the Markov Reference
Model (MRM), already suggested by \cite{MattsonGST70}, where the
address trace is a finite Markov chain.  A wealth of results are
obtained by Karlin, Phillips, and Raghavan \cite{KarlinPR00} for MRM,
including the Commute Algorithm, a remarkable policy computable from
the transition probabilities of the chain in polynomial time, whose
expected miss rate is within a constant factor of optimum. We
underscore that MRM and LRUSM are substantially different models; in
general, while the MRM trace is itself a Markov process with $V$
states, the LRUSM trace is a function of a Markov process with $V!$
states.

\paragraph{Paper outline} In this work, we further the study of the
LRUSM, deriving new results and strengthening its understanding. Some
methodological aspects of our investigation, however, can be of
interest for a wider class of models of the reference trace, hence
they will be presented in a more general context.

A first methodological aspect we explore is the possibility as well as
the fruitfulness of casting miss minimization as a problem of optimal
control theory (or, equivalently, as a Markov Decision Problem
\cite{LewisP02}).  In Section \ref{sec:ocfr}, we show as this is
possible whenever the trace is a hidden Markov process, a scenario
which we call the Hidden Markov Reference Model (HMRM). The IRM, the
MRM, and the LRUSM are all special cases of the HMRM.  We refer to the
classical optimal control theory framework as presented, for instance,
in the textbook of Bertsekas \cite{Bertsekas00}.  In the dynamical
system to be controlled, the state encodes the content of the buffer
and some information on the past trace.  The disturbance input models
the uncertainty in the address trace, while the control input encodes
the eviction decisions available to the memory manager. The cost per
step is one if a miss is incurred and zero otherwise.  We modify the
standard assumption that the control is a function only of the state
and allow it to depend on the disturbance as well:
$u_t=\mu_t(x_t,w_t)$.  This modification is necessary since eviction
decisions are actually taken with full knowledge of the current
access.  The Bellman equation characterizing the optimal policies has
to be modified accordingly. The technicalities of this adaptation are
dealt with in the Appendix.

A second methodological aspect we explore is a generalization of the
buffer management problem where rather than imposing the capacity as a
fixed constraint, we let the number of buffer positions vary
dynamically, under the control of the management policy. The average
buffer occupancy becomes a second cost, in addition to the miss rate,
and the tradeoff between the two costs is of interest. This problem
could have practical applications. For example, in a time-shared
environment, processes could be charged for their average use of main
memory and hence be interested in utilizing more memory in phases
where this can result in significant page-fault reduction and less
memory in other phases. In addition, the study of the average
occupancy problem sheds significant light even on the solution of the
fixed capacity problem. In particular, for the LRUSM, an optimal
policy for the case of a fixed buffer can be simply obtained from
optimal policies for the average occupancy case. In Section
\ref{sec:relax-buff-manag}, the average buffer occupancy problem is
studied for the HMRM, exploiting techniques of multi-objective
optimization and optimal control. A key advantage lies in the possibility of composing a
global policy from policies tailored to the individual items in the
virtual space. Furthermore, the approach naturally lends itself to the
efficient management of a buffer shared among different processes.
Throughout this section, the notion of optimality of the policies
under consideration is that of \emph{gain} optimality over an infinite
horizon.

In Section \ref{lrusm}, the framework and the results developed in the
two previous sections are applied to the LRUSM. After reviewing the
LRU stack model, a class of optimal policies is derived, for an
arbitrary distribution $s$, for the average occupancy problem. It is
then shown that this class of policies includes some that use a buffer
of fixed capacity $C$.  More specifically, it is shown that the
minimum miss rate under a constraint on the average occupancy can be
attained with fixed capacity. A buffer of capacity $C$ is optimally
managed by a $\kl$ policy \cite{WoodFL77,WoodFL83}, which can be
specified by two parameters, denoted as $K=K(C)$ and $L=L(C)$.  $\kl$
policies include, as special cases, LRU ($K=C$, $\forall L$) and Most
Recently Used (MRU) ($K=1$, $L=V$). Our derivation of the optimal
policy has a number of advantages. (i) The optimality is established
for the more general setting of average occupancy. (ii) The policy is
naturally described in terms of a system of priorities for the
eviction of items. As a corollary of a result of \cite{MattsonGST70},
it follows that the policy does satisfy the inclusion property: if an
item is in a given buffer, then it is also in all buffers of larger
capacity.  The inclusion property rules out the so-called Belady
anomaly
\cite{BeladyNS69} and enables more efficient algorithms for its
performance evaluation.  (iii) By linking the eviction priorities to
the (planar) convex hull of the Pareto optimal points of the average
occupancy problem and by adapting Graham's scan, an algorithm is
derived for a \emph{linear} time computation of the values $K(C)$ and
$L(C)$ for all relevant values of $C$. (Previously known properties of
the $\kl$ policy lead to a straightforward cubic algorithm.) (iv)
Finally, the priorities can be shown to correspond to a suitable
notion of profit rate of an item, informally capturing the best
achievable ratio between expected hits and the expected occupancy for
that item, leading to the concept of Least Profit Rate (LPR)
policy. The LPR policy can be defined for models different from the
LRUSM for which, while not necessarily optimal, it may yield a good
heuristic.

In Section~\ref{sec:olol}, we show that the ratio $\chi$
between the expected miss rate of the optimal on-line policy, LPR, and
that of OPT is $O(\log C)$. Moreover, for the class of stack access
distributions $s$ for which the miss rate of LPR is lower bounded by
some constant $\beta>1/C$, we have $\chi \leq 2\ln(2/\beta) \in O(1)$.

The ability to efficiently compute the number of misses for buffers of
various capacities when adopting a given policy for benchmark traces
is of key interest in the design of hardware as well as software
solutions for memory management
\cite{BilardiEP11,SugumarA93,ThompsonS89}.
In Section~\ref{sec:fast-simul-optim}, we develop an algorithm to compute the
LPR misses for all buffer capacities in time $O(\log V)$ per access, providing a
rather non trivial generalization of an analogous result for LRU
\cite{BennettK75,AlmasiCP02}.

In the remainder of the paper, we explore alternate notions of
optimality. In Section~\ref{sec:fh}, we consider optimization over a
finite interval, or horizon.  Technically, the optimal control problem
is considerably harder. As an indication, even if the system dynamics,
its cost function, and the statistics of the disturbance are all time
invariant, the optimal control policy is in general time-dependent.
We show that, for any monotonically non increasing stack distribution
$s$, LRU is an optimal policy for any finite horizon, whereas MRU is
optimal if the stack distribution is non decreasing.  While these
results appear symmetrical and highly intuitive, their proofs are
substantially different and all but straightforward.  The standard
approach based on the Bellman equation, which requires ``guessing''
the optimal cost as a function of the initial state, does not seem
applicable, lacking a closed form for such function. We have
circumvented this obstacle by establishing an inductive invariant on
the relative values of the cost for select pairs of states. This
approach may have applicability to other optimal-control problems.
Some of the results are derived for a considerably more general
version of the LRUSM, which does not assume the statistical
independence of stack distances at different steps.

Finally, in Section~\ref{sec:ih}, we take a preliminary look at 
\emph{bias} optimality, a property stronger than \emph{gain} optimality,
but also considerably more difficult to deal with. As an indication of
the obstacles to be faced, we prove that, in some simple cases of
LRUSM, no bias-optimal policy satisfies the useful inclusion property.
We also develop a closed-form solution for the simplest non-trivial
case of buffer capacity, that is, $C=2$.  The derivation as well as
the results are not completely straightforward, suggesting that the
solution for arbitrary buffer capacity may require considerable
ingenuity.

We conclude the paper with a brief discussion of directions for
further research.

The main notation introduced and used throughout this work is summarized in
Table~\ref{tab:notation}.

\newcommand{\otoprule}{\specialrule{\heavyrulewidth}{.5em}{.5em}}
\begin{table}
  \centering
  \begin{tabular}{lp{.8\textwidth}}
    \toprule
    \textbf{Symbol} & \textbf{Description}\\
    \otoprule
    $\opt$ & The (off-line) optimal replacement policy\\
    $\lru$ & The Least Recently Used replacement policy\\
    $\lpr$ & The Least Profit Rate replacement policy\\
    $V$ & Size of the high-latency memory\\
    $C$ & Size of the low-latency buffer\\
    $a_t$ & Address of the item accessed at time $t$\\
    $\Lambda_t, d_t$ & The $\lru$ stack at time $t$ and $\lru$ stack depth of
    item accessed at time $t+1$ ($\Lambda_t(d_t)=a_{t+1}$)\\
    $s(j)$ & Probability of accessing the $j$-th most recently referenced item\\
    $S(j)$ & Cumulative probability distribution of $s$: $S(j) = \sum_{i=1}^j s(i)$\\
    $\chi$ & Stochastic competitive ratio\\
    $x_t$ & State at time $t$ (dynamical systems)\\
    $w_t$ & Disturbance at time $t$ (dynamical systems)\\
    $u_t=\mu_t(x_t, w_t)$ & Control at time $t$ (dynamical systems)\\
    $x_{t+1}=f(x_t,u_t,w_t)$ & State transition (dynamical systems)\\
    $g(x_t, w_t)$ & Cost (dynamical systems)\\
    $J^*_t(x_0)$ & Optimal cost starting from state $x_0$ with a time horizon of
    $t$ steps (dynamical systems)\\
    \bottomrule
  \end{tabular}
  \caption{Main notation used throughout the paper}
  \label{tab:notation}
\end{table}

\section{Optimal Control Formulation of Eviction for the
Hidden Markov Reference Model}
\label{sec:ocfr}

In the typical problem of optimal control \cite{Bertsekas00}, one is given a
dynamical system described by a state-transition equation of the form
\begin{equation}\label{eqn:xtransition}
  x_{t+1}=f(x_t,u_t,w_t) \mc
\end{equation} 
where $x_t$ is the \emph{state} at time $t$, while both $u_t$ and $w_t$ are
inputs, with crucially different roles. Input $u_t$, called the \emph{control},
can be chosen by whoever operates the system.  In contrast input $w_t$,
historically called the \emph{disturbance}, is determined by the environment and
modeled as a stochastic process. At each step $t$, a cost is incurred, given by
some function $g(x_t,u_t,w_t)$. The objective of optimal control is to find a
\emph{control policy} $u_t=\mu_t(x_t)$ so as to minimize the total cost
\begin{equation}
  \E_w \left[ \sum_{t \in I} g(x_t,u_t,w_t) \right] \mc
\end{equation} 
where $I$ is a time interval of interest.  A key premise of most optimal control
theory is the assumption of \emph{past-independent disturbances} (PID): given the
current state $x_t$, the current disturbance $w_t$ is statistically independent
of past disturbances $\{w_{\tau}:~\tau <t\}$.

In this section, we define a dynamical system whose optimal control
corresponds to the minimization of the number of misses when the
reference trace can be expressed as a function $a_t=r(z_t)$ of a
Markov chain $z_t$ with a finite state space $Z$.  We call this
scenario the Hidden Markov Reference Model (HMRM).
In the following we assume the system to be \emph{unichain}, i.e., under any
stationary policy the Markov chain associated with the system evolution has only
a single recurrent class. This hypothesis guarantees that the average cost in
infinite horizon does not depend on the initial state and that there is always a
solution to the Bellman equation (e.g., a Markov chain with two non
communicating classes is not unichain).

To cast eviction as a problem in optimal control, the state of our
dynamical system will model both the Markov chain underlying the trace
and the content of the buffer:
\begin{equation}
  x_t \deq (z_t, b_t) \mc
\end{equation}
where $b_t\in\{0,1\}^V$ is a Boolean vector such that, for $j=1,
\ldots, V$,
\begin{equation}
  b_t(j) \deq
  \begin{cases}
    1 & j \text{ is in the buffer},\\
    0 & \text{otherwise}.
  \end{cases}
\end{equation}

Toward formulating the transition function governing the evolution of
$x_t$, let us first observe that any Markov chain can be written
as
\begin{equation}\label{eqn:ztransition}
  z_{t+1} = \phi (z_t, w_t) \mc
\end{equation}
where $w_t \in W$ is a sequence of equally distributed random variables,
independent of each other and of the initial state $z_0$ \cite{LevinPW09}.
Furthermore, $W$ is a finite set with $|W| \leq |Z|(|Z|-1)$.  We take $w_t$ to
be the ``disturbance'' in (\ref{eqn:xtransition}).  We let the control input
$u_t \in \{0,1,\ldots,V\}$ encode the eviction decisions with $0$ denoting no
eviction (the only admissible control in case of a hit) and $j>0$ denoting the
eviction of the item $j$ (an admissible control only when a miss occurs and the
item $j$ is in the buffer, i.e., $b_t(j)=1$). We can then write:
\begin{equation}\label{eqn:btransition}
  b_{t+1} = \psi (b_t, u_t) \mc
\end{equation}
where the transition function $\psi$ is specified as
\begin{equation}
  b_{t+1}(j) \deq
  \begin{cases}
    1 & j = a_t = r(z_t) \mc \\
    0 & j = u_t \mc\\
    b_t(j) & \text{otherwise} \mf
  \end{cases}
\end{equation}
Finally, we have:
\begin{equation}\label{eqn:Xtransition}
  x_{t+1}= (z_{t+1}, b_{t+1}) = (\phi (z_t, w_t), \psi (b_t, u_t)) 
         = f(x_t,u_t,w_t) \mf
\end{equation}
The instantaneous cost function $g$ is simply
\begin{equation}\label{eqn:gDef}
  g(x_t, w_t) \deq
  \begin{cases}
    1 & \text{if a miss occurred} \mc\\
    0 & \text{otherwise} \mf
  \end{cases}
\end{equation}
Finally, we assume that a policy can set the control $u_t$ with
knowledge of both the state and the disturbance: $u_t=\mu_t(x_t,w_t)$.
This requires adaptation of some results derived in the control theory
literature typically assuming $u_t=\mu_t(x_t)$.

Consider a policy $\pi=\left(\mu_1, \mu_2, \ldots, \mu_\tau\right)$
applied to our system during the time interval $\left[1,\tau \right]$,
so that, for $t$ in this interval, we have
\begin{equation}\label{mi-dynamics}
    x_{t+1}=f\left(x_t, \mu_t\left(x_t,w_t\right), w_t \right)\mf
\end{equation}
We define the cost of $\pi$, starting from state $x_0$, with time
horizon $\tau$ as
\begin{equation}\label{Jdef}
  J_\tau^\pi\left(x_0\right) \deq \E_{\vec w}\left[\sum_{t=0}^{\tau-1}
    g\left(x_t, w_t\right)\right], \quad \text{where } \vec w=\left[w_0, w_1,
    \ldots, w_{\tau-1} \right] \mc
\end{equation}
where the $x_t$'s are subject to (\ref{mi-dynamics}) and the expected value
averages over disturbances $w_t$. For our system, this is the expected number of
misses in $\tau$ steps. The optimal cost is
\begin{equation}
  J^*_\tau(x_0) \deq \min_\pi J_\tau^\pi\left(x_0\right) \mf
\end{equation}
The optimal cost satisfies the following dynamic-programming
recurrence (analogous to Eq.~1.6 in Vol.~1 of \cite{Bertsekas00})
\begin{equation}\label{beqsys1}
  \begin{split}
    J^*_\tau(x_0)=\E_w\left[\min_{u\in U(x_0,w)}\left\{g(
        x_0,w)+J^*_{\tau-1}\left(f\left(x_0,u,w\right)\right)\right\}\right]\mc
  \end{split}
\end{equation}
where $U(x_0,w)$ denotes the set of controls admissible when the state is $x_0$
and the disturbance is $w$.  Introducing a vector $\vec J^*$ whose components
are the values $J^*_{\tau}(x)$, in some chosen order of the states,
(\ref{beqsys1}) can be concisely rewritten as
\begin{equation}
  \vec J^*_\tau =\T \vec J^*_{\tau-1} \mc
\end{equation}
where $\T$ is the (non-linear) \emph{optimal-cost update operator}.

In applications where the temporal horizon of interest is long and perhaps not
known a priori, one is interested in policies that are optimal over an infinite
horizon; an added benefit is that such policies are provably \emph{stationary},
under very mild conditions.  Usually, the cost defined in (\ref{Jdef}) diverges
as $\tau \rightarrow \infty$, thus alternate definitions of optimality are
considered \cite{Bertsekas00,LewisP02,ArapostathisBFGM93,blackwell1962discrete},
like gain and bias optimality.
\begin{description}
\item[Gain Optimality] refers to a policy $\mu^*$ that achieves the lowest \emph{
    average cost} (that can be shown to be independent of $x$): $\mu^*= \arg
  \min_\mu \lim_{\tau\rightarrow\infty} {J^\mu_\tau(x)}/{\tau}$ .
\item[Bias Optimality] refers to a policy $\mu^*$ that is as good as any other
  stationary policy for $\tau$ long enough, i.e.,\ $\forall
  x\;\lim_{\tau\rightarrow\infty} J^\mu_\tau(x) - J^{\mu^*}_\tau(x)\geq
  0$. (Note that a bias optimal policy is also gain optimal.)
\end{description}

In this work we mostly concentrate on the standard gain optimality concept
(which corresponds to the miss rate minimization), but we also give some
particular results for the stronger concept of bias optimality in
Section~\ref{sec:ih}.  In Sections~\ref{sec:relax-buff-manag} and \ref{sec:ih}
we will use the classical \emph{Bellman equation}, which characterizes optimal
control policies in infinite horizon as solutions of a fixed point equation
\cite{Bertsekas00}. We show in \ref{sec:bellman-equation} that the result holds
in our model as well:
\begin{proposition}\label{bel}
  Let $\Delta$ be a dynamical system, with state space $X$, whose control $u_t$
  can be chosen with knowledge of the disturbance $w_t$. If
     $\exists \lambda \, \exists \vec h :\quad \lambda \vec 1 + \vec h = \T \vec h$,     then $\lambda$ is the \emph{optimal average cost} of $\Delta$ and
  $\vec h$ is the vector of the \emph{differential costs} of the
  states, i.e.
  \begin{equation}
    \forall x,y \in X \quad \lim_{\tau\rightarrow +\infty} J_\tau^*(x)-J_\tau^*(y)=h(x)-h(y)\mf
  \end{equation}
\end{proposition}
In the rest of the paper devoted to infinite horizon, since the costs will be
independent of the initial state $x_0$, we will simplify (and abuse a little)
the notation by setting $J^\mu(x_0)=J(\mu)$.

\section{The Average Occupancy Eviction Problem}
\label{sec:relax-buff-manag}

The classical form of the eviction problem, as reviewed in the
Introduction, is based on the assumption of a \emph{fixed capacity}
buffer: when the buffer is full, an eviction is required every time a
miss occurs, otherwise no eviction is performed. In this section, we
consider a relaxed version of the eviction problem where the buffer
is of potentially unlimited capacity ($C=V$ is actually sufficient)
and a variable portion of it can be occupied at different times. The
requirement that the item being referred must be kept in the buffer or
brought in if not already there is retained. However, after an access
(whether a hit or a miss), any item in the buffer, except for the one
just accessed, can be evicted. In this \emph{relaxed buffer
management} scenario, in addition to the miss rate, an interesting
cost metric is the \emph{average occupancy} of the buffer.  We call
average occupancy eviction problem the minimization of the miss
rate, given a target value for the average occupancy. This problem has
direct applications, as mentioned in the Introduction. However, the
study of average occupancy also sheds light on the fixed capacity
version, as we will see in particular in the next section for the
LRUSM.  A key advantage of the average capacity problem is that its
solution can be obtained by combining policies for the individual
items considered in isolation.

\subsection{The Single Item Problem}
\label{sec:one-item-eviction}

In this section we study eviction policies for single items from a
multiobjective perspective \cite{Ehrgott05,Miettinen99}, i.e., we are not
interested in optimizing a single objective, but in obtaining all the ``good''
policies.

When focusing on a single item, say $\omega$, a first simplification
arises from the fact that the state of the buffer, denoted
$\beta_\omega$, is just a binary variable, set to $1$ when the item is
kept in the buffer and to $0$ otherwise.  One also can restrict
attention to the $\omega$-\emph{trace} $a_{\omega,t}$, defined as $1$
when the item is referenced ($a_t=\omega$) and to $0$ otherwise.
Clearly, if the full trace $a_t$ is a hidden Markov process, so is
also the $\omega$-trace. However, in some cases, the $\omega$-trace
can be described with fewer states, forming a set $Z_\omega$. For
example, the reduction is dramatic for the LRUSM, where $|Z|=V!$ while
$|Z_\omega|=V$, for every item $\omega$. In general, we can write
\begin{equation}\label{eqn:phiOmega}
a_{\omega,t} = r_{\omega}(z_{\omega,t}), \quad \text{with }
z_{\omega,t+1}= \phi_\omega(z_{\omega,t}, w_{\omega,t}) \mf
\end{equation}
Process $z_{\omega,t}$ is called the \emph{Characteristic Generator}
(CG) of item $\omega$.  The control input $u_{\omega,t} \in \{0,1\}$
determines whether $\omega$ is evicted from the buffer ($u_{\omega,t}=1$),
which is admissible only when the item is in the buffer at time $t$
($\beta_{\omega,t}=1$) and it is not currently accessed
($a_{\omega,t}=0$). This amounts to specifying the function
$\psi_\omega$ such that
\begin{equation}\label{eqn:psiOmega}
\beta_{\omega,t+1} = \psi_\omega(\beta_{\omega,t}, u_{\omega,t}) \mf
\end{equation}
The overall state of the control system is then $x_\omega = (z_\omega,
\beta_\omega)$, evolving as
\begin{equation}\label{eqn:fOmega}
x_{\omega,t+1} = f_\omega(x_{\omega,t}, u_{\omega,t}, w_{\omega,t}) \mf
\end{equation}
We are interested into two types of cost: buffer occupancy and misses.
Correspondingly, we introduce two instantaneous cost functions:
\begin{align}\label{eqn:gOmega}
  g_{\oc, \omega}(x_{\omega,t}, w_{\omega,t})& = \beta_{\omega,t}&
  g_{\ms, \omega}(x_{\omega,t}, w_{\omega,t})& =
  \begin{cases}
    1, &\beta_{\omega,t}=0 \; \wedge \; a_{\omega,t}=1 \quad \text{(miss)} \\
    0, & \text{otherwise} \mf
  \end{cases}
\end{align}
Obviously, there is a tradeoff between the two costs, occupancy being
minimized by evicting the item as soon as possible and misses being minimized by
never evicting it. The set of ``good'' solutions to multiobjective optimization
problems is the Pareto set, i.e., the set of all policies whose costs are not
dominated by the costs of other policies (Pareto points are also known as
Efficient Points, EPs).

A peculiarity of multiobjective optimization problems is that, since we are
studying tradeoffs among the costs, sometimes we can usefully introduce
Randomized Mixtures of Policies (RMoPs) to obtain more points in the costs
space. E.g., consider the problem of choosing a route every day from home to
work between two possible routes $\mu_1$ and $\mu_2$, with costs $\vec J(\mu)$
defined by the driving time and gas used. If $\vec J(\mu_1)=(20 \text{
  minutes},10 \$)$ and $\vec J(\mu_2)=(30 \text{ minutes}, 5 \$)$ by choosing
every day with the same probability $\mu_1$ or $\mu_2$ we have long-run average
costs equal to $(25 \text{ minutes}, 7.50 \$)$, which cannot be obtained by
choosing always the same route. To properly define an RMoP for the average
occupancy problem for a single item, we observe that there always exists a
``hit'' state $x^*=(z^*,1)$ (with $r_\omega(z^*)=1$) which is recurrent in the
system evolution (otherwise the item frequency would be zero and there would not
be need to buffer it). If we call $t_1, t_2, \ldots$ the times at which the
system enters the recurrent state, then the problems of choosing an eviction
policy in the intervals $[t_i+1, t_{i+1}]$ are independent; by randomly choosing
for each interval between two policies we can obtain all the convex combinations
of the costs of the original non-mixture policies (see Fig.~\ref{fig:pareto}).
In more detail, if policies $\mu'$ and $\mu''$ exist with $J_{\oc}(\mu')=\eta'$
and $J_{\oc}(\mu'')=\eta''$ such that $C = \gamma \eta' + (1-\gamma) \eta''$ for
some $\gamma \in [0,1]$, we write $\mu = \rand_\gamma\left(\mu', \mu''\right)$
to mean that $\mu$ is an RMoP that chooses $\mu'$ with probability $\gamma$ and
$\mu''$ with probability $1-\gamma$ (the random choice being made every time the
system leaves the state $x^*$). Note that an analog equation holds also for the
miss rate cost: $J_{\ms}(\mu)=\gamma J_{\ms}(\mu') + (1-\gamma) J_{\ms}(\mu'')$.

\begin{figure}
  \centering
  \includegraphics{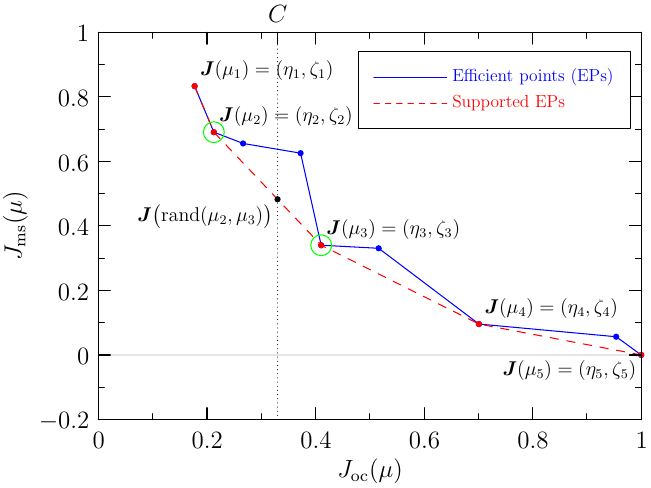}
  \caption{A Randomized Mixture of Policies (RMoP). Achieving $C$ as average
    buffer occupancy using a probabilistic combination of $\mu_2$ and $\mu_3$.}
  \label{fig:pareto}
\end{figure}

Since any point in the costs space between two policies can be obtained by an
appropriate RMoP, the set of ``good'' policies in this context maps into the set
$\cal{S}$ of Supported EPs (SEPs, also known as Pareto-convex points) which have
costs not dominated by any convex combination of the costs of other
policies. The set of SEPs can in general be of size exponential in $|Z_\omega|$
and thus it is interesting to find a polynomial approximation $\hat{\cal{S}}$ of
$\cal{S}$. More specifically, for every point $(J_1,J_2)$ in $\cal{S}$ we want
two points $({\hat J}_1', {\hat J}_2')$ and $({\hat J}_1'', {\hat J}_2'')$ in
$\hat{\cal{S}}$ and a real number $0 \leq \gamma \leq 1$ such that, for $q=1,2$,
we have $J_q (1+\epsilon) \geq \gamma {\hat J}_q' + (1 - \gamma) {\hat J}_q''$.

A necessary and sufficient condition \cite{PapadimitriouY00,DiakonikolasY08} to
get a Polynomial Time Approximation Scheme (PTAS) to construct $\hat{\cal{S}}$
is the availability of a PTAS for solving the single-objective scalarized
problem (\Pscal) which has the following cost:
\begin{equation}\label{eqn:psiGtheta}
G_{\omega, \theta}(x_{\omega}, w_{\omega})
= \cos(\theta) g_{\oc, \omega}(x_{\omega}, w_{\omega})
+ \sin(\theta) g_{\ms, \omega}(x_{\omega}, w_{\omega}) \mf
\end{equation}
The PTAS to build $\hat{\cal{S}}$ runs in time polynomial in both $|Z_\omega|$
and $(1/\epsilon)$.
(Note that any relative weight between the two costs can be achieved by a
suitable choice of $\theta \in \left[0, \pi/2\right]$.)
While \Pscal\ is interesting in its own right, here we are considering it mainly
as a tool to find a polynomial approximation of the set of SEPs.
For the single item problem the optimal policy solution of \Pscal\ can be
obtained by solving the following Bellman equation \cite{Bertsekas00}, in the
unknowns $\lambda_{\omega,\theta} \in \reals$ and $h_{\omega,\theta} : X_\omega
\rightarrow \reals$, where $X_\omega = Z_\omega \times \{0,1\}$:
\begin{equation}\label{eqn:bellmanGtheta}
  \lambda_{\omega,\theta} + h_{\omega,\theta}(x) = \E_{w_{\omega}} \left [G_{\omega,
      \theta}(x_{\omega}, w_{\omega}) +
    \min_{u_{\omega}} \left\{h_{\omega,\theta} \left(f_\omega(x_{\omega},
    u_{\omega}, w_{\omega})\right) \right\} \right] \mf
\end{equation}
The solution can be computed in time polynomial in $|Z_\omega|$ (the
problem is P-complete \cite{PapadimitriouT87}). The optimal policy can
be obtained as $\mu_{\omega,\theta}(x_{\omega}, w_{\omega}) =
\arg\min_{u_\omega} \left\{h_{\omega,\theta} \left(f_\omega(x_{\omega},
u_{\omega}, w_{\omega})\right) \right\}$.  From the general theory,
the cost of the optimal policy $\mu_{\omega,\theta}$ (for the
instantaneous cost $G_{\omega, \theta}$) is $J_{\omega,\theta}
=\lambda_{\omega,\theta}$.  Standard methods permit the polynomial
time computation of the values $J_{\oc,\omega}$ and $J_{\ms,\omega}$
of the same policy for the instantaneous costs $g_{\oc, \omega}$ and
$g_{\ms, \omega}$, respectively.

\subsection{The Average Occupancy of All Items}
\label{sec:aver-occup-all}

The solution for the average occupancy problem of all items can be obtained by
choosing, for each item $\omega$, an appropriate solution (given in general by
an RMoP) to the single item problem. In fact, let $\mu$ be an optimal policy for
the all items problem, then $\mu$ induces an occupancy cost
$J_{\oc,\omega}(\mu)$ and a miss rate cost $J_{\ms,\omega}(\mu)$ for each item
$\omega$. Now consider the single item problem of finding a policy $\mu_\omega$
which minimizes $J_{\ms,\omega}(\mu_\omega)$ while achieving occupancy
$J_{\oc,\omega}(\mu_\omega)=J_{\oc,\omega}(\mu)$: then we must necessarily have
$J_{\ms,\omega}(\mu_\omega)=J_{\ms,\omega}(\mu)$, otherwise either $\mu$ or
$\mu_\omega$ would be suboptimal.
Because of this relationship between single and all items optimal policies, the
average occupancy problem for all items reduces to a convex resource allocation
problem, which can be solved efficiently
\cite{StoneTW92,ThiebautSW92,FredericksonJ82}, once we have the approximation of
the set of SEPs for each single item.

A quantity of pivotal importance is the marginal gain $\rho_\omega$ of policy
$\mu_\omega'$ w.r.t.\ to policy $\mu_\omega$ defined, as
\begin{equation}
  \rho_\omega \deq -
  \frac{J_{\ms,\omega}(\mu_\omega')-J_{\ms,\omega}(\mu_\omega)}
  {J_{\oc,\omega}(\mu_\omega')-J_{\oc,\omega}(\mu_\omega)}
  \mf
\end{equation}
The solution can be obtained in a greedy fashion by increasingly allocating
buffer capacity to the item which provides a larger marginal gain. This greedy
procedure takes time polynomial in $\sum_\omega|Z_\omega|$; a more sophisticated,
asymptotically optimal algorithm can be found in \cite{FredericksonJ82}.  Note
that we can obtain the global miss rate $M$ and buffer occupancy $C$ just adding
up the equivalent quantities for the single items:
\begin{align}
  M &= \sum_\omega J_{\ms,\omega} \mc &
  C &= \sum_\omega J_{\oc,\omega} \mf
\end{align}

\begin{theorem}
  Let $\vec J(\mu^1_\omega)=(\eta^1_\omega,\zeta^1_\omega), \, \vec
  J(\mu^2_\omega)=(\eta^2_\omega,\zeta^2_\omega), \, \ldots$ be the list, with
  $\eta^i_\omega < \eta^{i+1}_\omega$, of the SEPs for the generic item
  $\omega$. The optimal policies $\mu_\omega$ (given for each item) are obtained
  applying Algorithm~\ref{alg:multiple-items}.
\end{theorem}
{\centering
  \begin{algorithm}
    $\forall \omega \; c[\omega] \leftarrow 1$ \;
    $ B \leftarrow 1 $ \;
    \While{$B<C$}{

      \tcp{choose the policy improvement with maximum marginal gain}
      $\displaystyle \phi \leftarrow \arg\max_\omega \left\{
      -\frac{\zeta_\omega^{c[\omega]+1}-{\zeta_\omega^{c[\omega]}}}{\eta_\omega^{c[\omega]+1}
        - {\eta_\omega^{c[\omega]}}}\right\}$ \;
      $B \leftarrow B + \eta_\phi^{c[\phi]+1} - {\eta_\phi^{c[\phi]}}$ \;
      $c[\phi] \leftarrow c[\phi]+1$ \;
    }
    $\forall \omega\not = \phi \quad \mu_\omega \leftarrow
    \mu_\omega^{c[\omega]} $ \;
    $ \gamma \leftarrow \text{such that } B-\sum_{\omega \not = \phi} \eta_\omega^{c[\omega]} = \gamma
    \eta_\phi^{c[\phi]-1}+(1-\gamma) \eta_\phi^{c[\phi]}$ \;
    $\mu_\phi \leftarrow \rand_\gamma\left(\mu_\phi^{c[\omega]-1},\mu_\phi^{c[\omega]} \right)$ \;
    \caption{Obtaining the optimal algorithm for multiple items.}
    \label{alg:multiple-items}
  \end{algorithm}
}
\begin{remark}
  Algorithm~\ref{alg:multiple-items} will produce a non-mixture
  eviction policy for each item except for one, for which an RMoP may
  be needed in order to use the entire buffer space assigned to
  it. Note that, given a target value $C$ for the buffer occupancy,
  there exist a global policy made only of single item non-mixture
  policies which has average buffer occupancy $C-1 < C' \leq C$ and is
  optimal for that value of the occupancy.
\end{remark}

\subsection{Buffer Partitioning}
\label{bp}
An interesting application of the previous algorithm arises when we want to
partition a buffer of capacity $C$ among $n$ independent processes, each
described by a different HMRM. We assume that the $i$-th process accesses a
private address space size $V_i$, at each step the $i$-th process has
probability $\pi_i$ ($\sum_{i=1}^n \pi_i=1$) to be the one to run.  We want to
determine the capacity $C_i$ to devote to process $i$ (with $\sum_{i=1}^n
C_i=C$) so as to minimize the global miss rate over an infinite temporal
horizon, under the hypothesis that each process is using the optimal eviction
policy.

It is straightforward to see that this problem is equivalent to building the
global policy from single items with the cost $J_{\ms}$ of the items accessed by
the $i$-th process rescaled by a factor $\pi_i$.

\section{Optimal Policies for the LRUSM}
\label{lrusm}

In this section, we focus on stack optimal policies for the LRU Stack
Model.  After first reviewing the LRUSM (\S\ref{sec:lru-stack-model}),
we derive the optimal policy in the average occupancy framework
(\S\ref{sec:aver-occup-probl}); due to the strong LRUSM structure many
simplifications apply to the general procedure built in the previous
section leading to a very efficient way of computing the optimal
policy (linear in $V$). In \S\ref{sec:fixed-capac-probl} we focus on
the fixed occupancy eviction problem and, by linking its solution to
the average occupancy one, we are able to provide a stack eviction
policy (Least Profit Rate, LPR) which is optimal for the
LRUSM. Furthermore, we can fully characterize LPR behavior using a
priority function based on a notion of profit which might be also of
interest for other memory reference models.

\subsection{The LRU Stack Model}
\label{sec:lru-stack-model}

Mattson et al.\ \cite{MattsonGST70} observed that a number of
eviction policies of interest, including LRU, MRU, and OPT, satisfy
the following property.

\begin{definition}\label{def:stackPolicy} Given an eviction policy $\mu$
defined for all buffer capacities, let $B_t^\mu(C)$ be the content of
the buffer of capacity $C$ at time $t$, after processing references
$a_1,\ldots,a_{t}$.
We say that the \emph{inclusion property} holds at time $t$ if, for any $C>1$,
$B_t^\mu(C-1) \subseteq B_t^\mu(C)$, with equality holding whenever the bigger
buffer is not full ($\left|B_t^\mu(C)\right| < C$).
We say that $\mu$ is a \emph{stack policy} if it satisfies the inclusion
property at all times for all address traces, assuming that inclusion holds for
the initial buffers $B_0^\mu(C)$, with $1 \leq C \leq V$ (this is trivially
verified if we assume, as we do, the initial buffers to be empty).
\end{definition}
The optimal on-line policy LPR, developed in this section, is a stack
policy.
Inclusion protects from Belady's anomaly \cite{BeladyNS69} (i.e.,
increasing the buffer capacity cannot lead to a worse miss ratio, as
can happen with, e.g., FIFO) and enables a compact representation of
the content of the buffers of all capacities by the \emph{stack of the
policy}, an array $\Lambda^\mu_t$ whose first $C$ components yield the
buffer of capacity $C$ as
\begin{equation}
B_t^\mu(C)= \left[ \Lambda^\mu_t(1), \ldots, \Lambda^\mu_t(C) \right] \mf
\end{equation}
The \emph{stack depth} $d_t$ of an access $a_{t+1}$ is defined as its
position in the policy stack at time $t$, so that
\begin{equation}
a_{t+1}=\Lambda^\mu_{t}(d_t) \mf
\end{equation} 
Upon an access of depth $d$, a buffer incurs a miss if and only if $C<d$. Thus,
computing the stack depth is an efficient way to simultaneously track the
performance of all buffer capacities.

\begin{figure}
  \centering  \includegraphics[width=52mm]{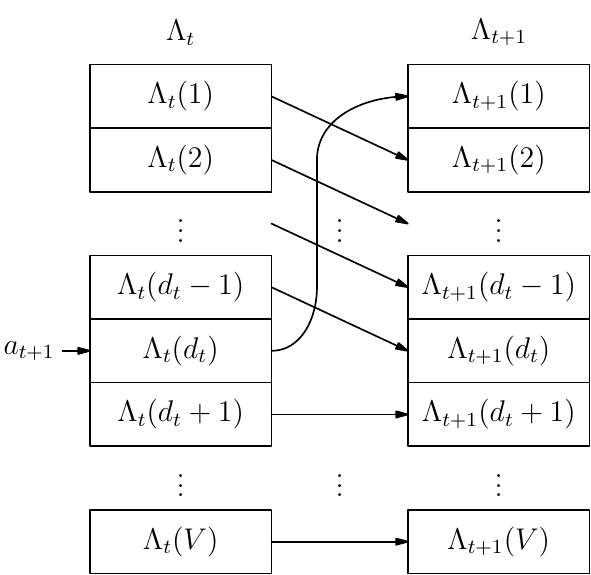}
  \caption{LRU stack update.}
  \label{lru-up}
\end{figure}

In the LRU stack, which we shall denote just by $\Lambda$ (i.e.,
$\Lambda=\Lambda^{\lru}$), the items are ordered according to the time
of their most recent access; in particular,
$\Lambda_{t+1}(1)=a_{t+1}$. Upon an access at depth $d_t$, the LRU
stack is updated by a downward, unit cyclic shift of its prefix of
length $d_t$, as illustrated in Fig.~\ref{lru-up}. The LRU stack has
inspired an attractive stochastic model for the address trace
\cite{OdenS72,SpirnD72,TurnerS77,EffelsbergH84,KobayashiM89}, where the
access depths $d_1,d_2, \ldots$ are \emph{independent and identically
distributed} random variables, specified by the distribution

\begin{equation}
  s(j) \deq \Pb[a_{t+1}=\Lambda_{t}(j)]=\Pr[d_t=j],
  \quad j=1, \ldots, V\enspace,
\end{equation}
or, equivalently, by the cumulative sum
\begin{equation}
  S(j) \deq \sum_{i=1}^j s(i), \quad j=1, \ldots, V \mf
\end{equation} 
We assume that an ``initial'' LRU stack is given. W.l.o.g., $s(V)\not = 0$
($s(V)=0$ implies that the last page in the stack is never accessed and can
therefore be ignored).

For example, the case where $s(j)$ decreases with $j$ captures a
strict form of temporal locality, where the probability of accessing
an item strictly decreases with the time elapsed from its most recent
reference.  It is simple to see that the actual trace $a_1,a_2,
\ldots$ can be uniquely recovered from the stack-depth sequence
$d_1,d_2, \ldots$, given the initial stack $\Lambda_0$.

To summarize, in the notation of Section~\ref{sec:ocfr}, the state
underlying the trace is the LRU stack ($z_t=\Lambda_t$) and the
disturbance is the stack distance ($w_t=d_t$). We have a HMRM, since
$a_{t+1}=r(\Lambda_{t+1})=\Lambda_{t+1}(1)$. The transition function $\phi$ for the trace
state (such that $\Lambda_{t+1} = \phi(\Lambda_t,d_t)$) corresponds to
the right unit cyclic shift of the the prefix of length $d_t$ of stack
$\Lambda_t$. The control input $u_t$, the buffer state $b_t$, and the
action of the former on the latter ($b_{t+1}=\phi(b_t,u_t) $) are
according to the general definitions given in  Section~\ref{sec:ocfr}.

\subsubsection{$\kl$ Eviction Policies}
\label{sec:an-optimal-stack}

A first optimal policy for the LRUSM was given by Wood, Fernandez and
Lang in \cite{WoodFL77,WoodFL83}. They introduced the
\emph{$\kl$ eviction policy}, defined as follows, in terms of LRU
stack distance, for given integers $K$ and $L$ that $1\leq K < C \leq
L \leq V$. If access $a_{t+1}$ results in a miss, then
\begin{itemize}
\item evict \(\Lambda_{t+1}(L+1)\), if it is in buffer $B_t(C)$;
\item otherwise evict \(\Lambda_{t+1}(K+1)\), which is always in
  $B_t(C)$ if the policy is consistently applied starting from an
  empty buffer.
\end{itemize}
Here, $\Lambda$ denotes the LRU stack, while $B(C)$ denotes the
content of the buffer under the $\kl$ policy.  Eviction is specified
in terms of the LRU stack immediately after the rotation due to access
$a_{t+1}$. Special cases are LRU ($K=C$, $\forall L$) and MRU ($K=1$,
$L=V$).  It has been shown in \cite{WoodFL77,WoodFL83} that, under the
LRUSM, for any $C$, there exist values $K(C)$ and $L(C)$ for which the
$\kl$ policy is (gain) optimal. It is easily shown that, in the steady
state, the items in the top $K$ positions of the LRU stack are always
in the buffer, the items in the bottom $V-L$ positions are always
outside the buffer, and $C-K$ of the $L-K$ items between position
$K+1$ and $L$ are in the buffer. (Technically, the $\kl$ policy is not
specified for a buffer containing neither \(\Lambda_{t+1}(L+1)\) nor
\(\Lambda_{t+1}(K+1)\). However, if the least recently used item is
evicted in such case, one can show that the same steady state is
reached, with probability 1.) It turns out that the miss rate of the
$\kl$ policy for a given distribution $s$ is
\begin{equation}
    M^{\kl}[s]= 1 - \frac{S(K)(L-C)+S(L)(C-K)}{L-K} \mf
\end{equation}
Finding, for each $C$, the optimal parameters $K(C)$ and $L(C)$ can be then
accomplished in time $\Theta(V^3)$, by evaluating the miss ratio for all
possible $(K,L)$ pairs.  In general, for a given $C$, the miss rate can be
minimized by different $(K,L)$ pairs. Smaragdakis et al.\
\cite{SmaragdakisKW99,SmaragdakisKW03} proved that, for any given $s$, $K(C)$
and $L(C)$ can be chosen so that the resulting $\kl$ policy does satisfy the
inclusion property.

In the next subsection, we derive the optimal eviction policy in the
average occupancy framework developed in
Section~\ref{sec:relax-buff-manag} and highlight its deep
structure. In \S\ref{sec:fixed-capac-probl} we derive a stack policy
which is optimal in the fixed buffer model. By exploiting its
similarity with the average occupancy optimal, we develop an algorithm
that computes the two $K$ and $L$ parameters for all $C \in \{1, V\}$
in \emph{linear} time, whereas a straightforward computation that does
not exploit the stack policy characterization takes \emph{cubic} time
(a first speedup from cubic to quadratic is obtained exploiting the
inclusion property, while the improvement from quadratic to linear
derives from algorithmic refinements).

\subsection{Average Occupancy Problem}
\label{sec:aver-occup-probl}

The general procedure described in the previous section can be
specialized for the LRUSM as follows. For each item $\omega$ the
associated CG is a Markov chain $z_{\omega,t}$ with states
$Z_\omega=\{1, \ldots, V \}$. The state encodes the position of the
item in the LRU stack, thus satisfying the following transition
probabilities:
\begin{align}
  \forall i \quad &\Pb[z_{\omega,t+1}=1 | z_{\omega,t}=i]=s(i) \\
  \forall i \neq 1 \quad &\Pb[z_{\omega,t+1}=i | z_{\omega,t}=i]=S(i-1) \\
  \forall i \neq V \quad &\Pb[z_{\omega,t+1}=i+1 | z_{\omega,t}=i]=1-S(i) \mf
\end{align}
We also have $r_\omega(1)=1$ and, $\forall i \neq 1, \quad r_\omega(i)=0$.  Thus, in the
LRUSM the statistical description of the CGs is the same for all the items. The
stationary eviction policies for the single item are binary vectors of size $V$
which say, for each LRU stack depth, if an item arriving at that position should
be evicted.
The ``hit'' state $x^* \deq (z^* = 1,\beta=1)$, in which the item is
at the top of the LRU stack and in the buffer, is recurrent and hence
will be used to define RMoPs in the LRUSM: we assume an RMoP will
choose a SEP policy each time the system leaves the state $x^*$. In
the following we call \emph{lifetime} the time between two consecutive
$x^*$. A lifetime ends when the item is accessed. Thus, at the
beginning of a new lifetime, an item is always at the top of the stack
and in the buffer. If evicted from the buffer, the item will re-enter
only at the beginning of the next lifetime.

We define $\ev_k$ as the policy which keeps the item in the buffer while at
stack depth $i\leq k$ and evicts it if $i > k$. Note that, in steady state, any
policy $\mu$ is equivalent to an appropriate $\ev_k$, where $k+1$ is the
smallest depth at which the policy $\mu$ evicts (this is because, after the item
is first accessed, it will always get evicted at depth $k+1$; it can only get
evicted once at a different depth). Hence, w.l.o.g., we can limit our study to
$\ev_k$ policies.
To provide a close form for both the occupancy and the miss cost of
$\ev_k$, we first need to establish some properties of the Markov
chain $z_\omega$.

\begin{proposition}
  \label{thm:time-spent-lrusm}
  Let $t^j_{j+h}$ be the expected time spent by an item in position $j+h$ ($h\geq
  0$) conditioned to the fact the item surely arrives in position $j$ without
  getting accessed. Then
  \begin{equation}
    t^j_{j+h} = \frac{1}{1-S(j-1)} \mf
  \end{equation}
  (Note that the expected time does not depend on $h$.)
\end{proposition}
\begin{proof}  For $h=0$ we easily have that
  \begin{equation}
    t_j^j=1+S(j-1) t_j^j \quad \Rightarrow \quad t_j^j=\frac{1}{1-S(j-1)}  \mf
  \end{equation}
  The event that an item starting from position $j$ does not arrive in position
  $j+1$ is the disjoint union of the events that there are exactly $h$ accesses
  with depth smaller than $j$ followed by one at depth $j$. So the probability
  for an item in position $j$ to arrive to position $j+1$ is
  \begin{equation}
    P_{j+1}^j=1-s(j)\sum_{h=0}^{+\infty}\left[S(j-1)\right]^h=\frac{1-S(j)}{1-S(j-1)}  \mc
  \end{equation}
  and hence the expected time spent in position $j+1$ is
  \begin{equation}
    t_{j+1}^j=P_{j+1}^j
    t_{j+1}^{j+1}=\frac{1-S(j)}{1-S(j-1)}\frac{1}{1-S(j)}=t_j^j  \mf
  \end{equation}
  Furthermore
  \begin{equation}
    t_{j+h}^j=P_{j+1}^j P_{j+2}^{j+1}\cdots P_{j+h}^{j+h-1}t_{j+h}^{j+h}=t_j^j  \mf
  \end{equation}
\end{proof}

As a special case, if we know that an item starts from position 1 (in
which it is going to spend exactly one time step in the current lifetime)
then it is going to spend an expected time $t_j^1=1$ in each position on
the LRU stack in its lifetime and hence a lifetime has an expected
length of $V$ timesteps. Under policy $\ev_k$, the item is buffered on
average for $k$ timesteps per lifetime and hence the occupancy cost is
given by
\begin{equation}
  \label{eq:coc}
  J_{\oc}(\ev_k)=\frac{k}{V} \mf
\end{equation}
As for the miss rate, since the each position at depth $j>k$
contributes with a probability $s(j)$ to the misses, we have that
\begin{align}
  \label{eq:cms}
  J_{\ms}(\ev_k) &= \frac{1-S(k)}{V} \mf
\end{align}
A direct application of Algorithm~\ref{alg:multiple-items} would
provide optimal policies $\mu_\omega$ identical for all the items
except for on one (which, in general, will be managed by a
RMoP). Actually, we can take a slightly different approach to exploit
the CGs symmetry: by forcing the eviction policy to be the same for
exactly \emph{all} the items the global problem is reduced to the
single-item one, with just a rescale of the costs $J_{\oc}$ and
$J_{\ms}$ by a factor of $V$.  The optimal policy is then an RMoP (the
same for all the items) which mixes two eviction points corresponding
to SEPs, causing each item to be evicted only at two possible stack
depths. The next theorem summarizes the preceding discussion.
\begin{theorem}\label{thm:aveCap}
Let $q_1=1 < q_2 < \ldots < q_l=V$ be the values of $k$ such that
$\ev_k$ is a SEP (i.e., a Pareto-convex) policy and let $q_i<C' \leq
q_{i+1}$, so that $C'=\gamma' q_i + (1-\gamma')q_{i+1}$ for some
$\gamma' \in (0,1]$. Then, the RMoP that, for each item, mixes
policies $\ev_{q_i}$ and $\ev_{q_{i+1}}$ with probability $\gamma'$
and $(1-\gamma')$, respectively, achieves optimal miss rate for
average occupancy $C'$.
\end{theorem}
It is a simple exercise to show that, if $s$ is monotonically
decreasing, then $l=V$ and $q_i=i$.  Instead, if $s$ is monotonically
increasing, then $l=2$, $q_1=1$, and $q_2=V$.

\subsection{Fixed Occupancy Problem}
\label{sec:fixed-capac-probl}

We have just shown that the optimal eviction policy in the average
occupancy model is characterized by two eviction points in the LRU
stack (let us call them $K'(C)$ and $L'(C)$). It is natural to
investigate what happens by applying a $\kl$ policy in the fixed
occupancy model with $K=K'(C)$ and $L=L'(C)$.  Under the $\kl$ policy
each lifetime is managed either by policy $\ev_K$ or $\ev_L$ and thus
its occupancy $C$ and miss rate $M$ can be written as
\begin{align}
  C &=\gamma J_{\oc}(\ev_K)+(1-\gamma)J_{\oc}(\ev_L) & 
  M &=\gamma J_{\ms}(\ev_K)+(1-\gamma)J_{\ms}(\ev_L) \mf
\end{align}
A set of similar equations holds for the average occupancy setting:
\begin{align}
  C' &=\gamma' J_{\oc}(\ev_{K'})+(1-\gamma')J_{\oc}(\ev_{L'}) & 
  M' &=\gamma' J_{\ms}(\ev_{K'})+(1-\gamma')J_{\ms}(\ev_{L'}) \mf
\end{align}
Since $C=C'$, $K=K'$ and $L=L'$ we must have the $\gamma'=\gamma$.
Then, we also have $M=M'$, i.e., the miss rate achieved by the $\kl$
policy for $K=K'(C)$ and $L=L'(C)$ is the same achieved by the average
occupancy optimal. Since the optimal miss rate for average occupancy
$C$ is obviously a lower bound for miss rate under fixed occupancy
$C$, the choice $K=K'(C)$ and $L=L'(C)$ parameters yields an optimal
policy under fixed capacity.

Based on these observations and on Theorem~\ref{thm:aveCap}, we can
compute, for each positive integer $C \leq V$, values $K(C)$ and
$L(C)$ for a gain optimal, fixed capacity policy.

\begin{corollary}\label{thm-KL-2-comp}
  Given a stack-distance distribution $s$, the corresponding values $q_1,q_2,
  \ldots, q_l$ can be computed in time \(\Theta(V)\) by
  Algorithm~\ref{alg-KL-2}, a specialization of the Graham scan \cite{Graham72}
  to obtain the convex hall of planar sets of points.  The same algorithm also
  computes optimal \(K(C)\) and \(L(C)\) for all (integer) buffer capacities $1
  \leq C \leq V$, uniquely determined from the $q_i$'s by the relation
\begin{equation}
  K(C)=q_i < C \leq q_{i+1}=L(C) \mf
\end{equation}
\end{corollary}
\begin{remark}
  Because of costs (\ref{eq:coc}) and (\ref{eq:cms}) finding the SEPs is
  equivalent to find the convex hull of the set $\{(1,0)\} \cup \{(j, S(j))\}_j
  \cup \{(V,0)\}$. Since these points are already ordered along the first
  coordinate we can directly apply the linear phase of Graham scan, skipping the
  initial sorting. Finally, being the points equally spaced along the first
  coordinate, the algorithm reduces to finding, for each point $j$, the next point $i$
  which maximizes the difference quotient:
  \begin{equation}
    \frac{S(i)-S(j)}{j-i+1} = \frac{\sum_{k=j}^i s(k)}{j-i+1} \mf
  \end{equation}
  This particular specialization of the Graham scan is studied in more detail in
  \cite{BernholtEH07,LinJC02}.
\end{remark}
\begin{algorithm}[H]
  \tcp{Graham scan specialization}
  \SetKw{KwDTo}{downto}
  \SetKw{KwPR}{print}
  $\nu[1] \leftarrow 1$; $\quad$
  $\pi[1] \leftarrow s(1)$; $\quad$
  $\Delta[1] \leftarrow 1$\;
  $\pi[V+1] \leftarrow 0$; $\quad$
  $\Delta[V+1] \leftarrow 1$\;
  \For{$j \leftarrow V$ \KwDTo $2$}{
    $\nu[j] \leftarrow j$; $\quad$
    $\pi[j] \leftarrow s[j]$\;
    $\Delta[j] \leftarrow 1$; $\quad$
    $n \leftarrow \nu[j]+1$\;
    \While{$\pi[j]/\Delta[j] \leq \pi[n]/\Delta[n]$}{
      $\nu[j] \leftarrow \nu[n]$\;
      $\pi[j] \leftarrow \pi[j]+\pi[n]$\;
      $\Delta[j] \leftarrow \Delta[j]+\Delta[n]$\;
      $n \leftarrow \nu[j]+1$\;
    }
  }
  \tcp{Print segmentation and priorities}
  $j \leftarrow 1$; $\quad$
  $i \leftarrow 1$\;
  \While{$j \leq V$}{
    \KwPR ``$q_i = \;$'' , $\nu[j]$\;
    $j \leftarrow \nu[j]+1$; $\quad$
    $i \leftarrow i+1$\;
  }
  $j \leftarrow 2$\;
  \While{$j \leq V$}{
    \KwPR ``$\xi(j) = \;$'' , $\pi[j]/\Delta[j]$\;
    $j \leftarrow j+1$\;
  }
  \caption{Computing \(K(C)\) and \(L(C)\) and Profit Rates priorities -- Linear
    algorithm. $K(C)=q_i < C \leq q_{i+1}=L(C)$ for some $i$.}
  \label{alg-KL-2}
\end{algorithm}

\subsubsection{The Least Profit Rate Eviction Policy}
\label{sec:least-profit-rate}

In this subsection we define a new stack policy, called Least Profit Rate
($\lpr$), by means of suitable priorities. We will see that, for any $C$, in
steady state, the LPR policy becomes the same as the $K(C)$-$L(C)$ policy, and
is therefore optimal.

\begin{definition}
  We define \(\bar s(i,j)\) as the average of \(s\) between \(i\) and \(j\):
  \begin{equation}
    \bar s(i,j) \deq \sum_{k=i}^j \frac{s(k)}{j-i+1} \mf
  \end{equation}
  Furthermore $\bar s(i)$ will indicate the moving average of $s$:
  \begin{equation}
    \bar s(i) \deq \bar s(1, i) \mf    
  \end{equation}
\end{definition}
\begin{definition}
  Let $\omega$ be an item identified by its stack depth $i$ (i.e.,
  $\Lambda(i)=\omega$). We define its profit rate $\xi$ as
  \begin{equation}
    \forall i\neq 1 \quad \xi(i) \deq \max_j \bar s(i,j) \mf
  \end{equation}
  (The profit rate of the last accessed item $\Lambda(1)$ is not defined, since
  it cannot be evicted.)
\end{definition}
\begin{definition}[Least Profit Rate]
  The $\lpr$ policy evicts the item $i^*$ in the buffer such that
  \begin{equation}
    i^* \deq \arg \min_i \xi(i) = \arg \min_i \max_j \bar s(i,j) \mf
  \end{equation}
  (If more items achieve the minimum, then the closest to the top of
   the stack is chosen.)
\end{definition}

Next, we can state the main results of this section:
\begin{theorem}\label{linathm}
  The \(K(C)\)-\(L(C)\) eviction policies (for the various $C$) are
  uniquely determined by the LPR (whose definition does not depend on $C$).
\end{theorem}

\begin{lemma}
  \label{sub-lemma}
  Let \(s\) be a function from natural to positive real numbers: \(s: \mathbb{N}
  \rightarrow \mathbb{R}^+ \).  Let \(\lambda\) be the point of maximum for
  \(\bar s\).  Then for each \(q \leq \lambda\) we have
  \begin{equation}
    \bar s(\lambda) \leq \bar s(q,\lambda) \mf
  \end{equation}
\end{lemma}
\begin{proof}
  We can rewrite \(\bar s(\lambda)\) as
  \begin{align}
    \bar s(\lambda)=&\sum_{j=1}^{\lambda} \frac{s(j)}{\lambda}=\frac{q-1}{\lambda}\sum_{j=1}^{q-1}
    \frac{s(j)}{q-1} \\&+ \frac{\lambda-q+1}{\lambda} \sum_{j=q}^{\lambda}
    \frac{s(j)}{\lambda-q+1}=\\
    =&\frac{q-1}{\lambda} \bar s(q) + \frac{\lambda-q+1}{\lambda} \bar
    s(q,\lambda)  \mf
  \end{align}
  Being \(\bar s(\lambda)\) a convex combination of \(\bar s(q)\) and \(\bar
  s(q,\lambda)\) the following inequality holds:
  \begin{equation}
  \min \left\{\bar s(q),\bar s(q,\lambda) \right\} \leq \bar s(\lambda) \leq \max
  \left\{\bar s(q),\bar s(q,\lambda) \right\} \mf
  \end{equation}
  Since for hypothesis we have \(\bar s(\lambda) \geq \bar s(q)\) we must have
  \begin{equation}
  \bar s(q) \leq \bar s(\lambda) \leq \bar s(q,\lambda)  \mf
  \end{equation}
\end{proof}

\begin{lemma}
  \label{sub-lemma-2}
  Let \(s\) be a function from natural to positive real numbers: \(s: \mathbb{N}
  \rightarrow \mathbb{R}^+ \).  Let \(\lambda\) be the point of maximum for
  \(\bar s\).  Then for each \(r>\lambda\) we have
  \begin{equation}
  \bar s(\lambda+1,r) \leq \bar s(\lambda) \mf
  \end{equation}
\end{lemma}
\begin{proof}
  We have the following convex combination:
  \begin{equation}
  \bar s(r)= \frac{\lambda}{r} \bar s(\lambda) + \frac{r-\lambda}{r} \bar s(\lambda+1,r) \mc
  \end{equation}
  with \(\bar s(\lambda) \geq \bar s(r)\) so we must have
  \begin{equation}
  \bar s(\lambda+1,r) \leq \bar s(r) \leq \bar s(\lambda) \mf
  \end{equation}
\end{proof}
\begin{remark}
  Lemmas \ref{sub-lemma} and \ref{sub-lemma-2} highlight the following
  useful properties of the sequence $(q_1,q_2, \ldots, q_l)$ obtained by
  Algorithm~\ref{alg-KL-2}:
  \begin{enumerate}
  \item The average value of $s$ is strictly decreasing between segments: $\bar
    s(1+q_i,q_{i+1}) > \bar s(1+q_{i+1},q_{i+2})$
  \item Within each segment the average of any prefix is smaller or equal to
    that of any suffix: $\forall k \in [1+q_i,q_{i+1}]$, then $\bar s(1+q_i,k) \leq
    \bar s(1+q_i,q_{i+1}) \leq \bar s(k,q_{i+1})$
  \end{enumerate}
\end{remark}
\begin{proof}[{\bf Proof of Thm.~\ref{linathm}}]
  Starting with an empty buffer we evict the first time when the top $C$
  position of the LRU stack are filled. The position $i^* = \arg \min_i \max_j
  \bar s(i,j)$ is exactly $K+1$ (as can be seen applying Lemmas~\ref{sub-lemma}
  and \ref{sub-lemma-2}). The following position $l$ that has $\max_j
  \bar s(l,j)< \max_j \bar s(i^*,j)$ is $L+1$, so each time an item reaches that
  position it is evicted (it can reach it only after a miss, since the positions
  after $L$ are not in the buffer).
\end{proof}

Below, we give a general formulation of the concepts of profit and of
\emph{profit rate} in the HMRM. For the definition we adopt the single item
average occupancy model. Given an item $\omega$ and a time $t$,
consider a single item eviction policy $\mu$ to determine, for any
underlying state $z$ of the CG, whether $\omega$ is kept in the buffer
or evicted upon reaching that state.  We call $\mu$-profit of $\omega$
the probability $\pi^{\mu}$ that $\omega$ is referenced before it is
evicted, which is a measure of how useful it would be to keep $\omega$
in the buffer under the policy. Let then $t+\Delta$, with $\Delta>0$,
be the earliest time after $t$ such that $\omega$ is either referenced
or evicted at time $t+\Delta$.  Clearly, $\E^{\mu}[\Delta]$ is a
measure of the storage investment made on $\omega$ to reap that
profit.  Therefore, the quantity $\pi^{\mu}/ \E^{\mu}[\Delta]$ is a
measure of profit per unit time, under policy $\mu$.  Finally, we call
\emph{profit rate} of $\omega$ the maximum profit rate achievable for
$\omega$, as a function of $\mu$.
(It ought to be clear that profits and profit rates depend upon the current
underlying state $z$ of the CG, although this dependence has not been reflected
in the notation, for simplicity.)

As for the LRUSM, let consider a buffered item at position $i$ in the
LRU stack. If we choose to keep it in the buffer until it goes past
position $j$, then the item is going to spend on average the same time
in each position between $i$ and $j$ and hence its profit per unit
time will be $\bar s(i,j)$.  This function has a maximum for some
value of $j$ that defines the item profit rate, as shown above.

The \emph{Least Profit Rate} (LPR) policy evicts, upon a miss, a page
in the buffer with minimum profit rate.  Profit rates \emph{are
independent of buffer capacity}, hence can be viewed as priorities. If
ties are resolved consistently for all buffer capacities, LPR
satisfies the inclusion property. In general, LPR is a reasonable
heuristic, but not necessarily an optimal policy in the fixed
occupancy model.

\section{On-line vs.\ Off-line Optimality}
\label{sec:olol}
Intuitively, the optimal off-line policy makes the best possible use
of the complete knowledge of the future address trace, whereas the
optimal on-line has only a statistical knowledge of the future. We
compare these two information conditions via the \emph{stochastic
competitive ratio}, defined as the following functional of the
distribution $s$:
\begin{align}
  \chi[s] \deq \frac{M^{\lpr}[s]}{M^{\opt}[s]}\mc
\end{align}
where $M^{\lpr}[s]$ and $M^{\opt}[s]$ denote the expected miss rates
(technically, the limit of the expected number of misses per step over
an interval of diverging duration, as considered in gain
optimality). The next theorem states the key result of this section.

\begin{theorem}\label{thm:compet}
  For any stack access distribution $s$ and buffer capacity $C$, $\chi[s]
  =O (\ln C)$.  If $M^{\lpr}[s] \geq \frac{1}{C}$, then the bound can be
    tightened as $\chi[s] \in O\left(\ln\frac{1}{M^{\lpr}[s]}\right)
  $.
\end{theorem}

To gain some perspective on the $O(\ln C)$ bound for LPR, we observe
that the stochastic competitive ratio of LRU can be as high as $C$,
(take $s : s(C+1)=1 $). We also remark that, for classes of
distributions where the miss rate of LPR is bounded from below by a
constant, the competitive ratio is bounded from above by a
corresponding constant.

Theorem ~\ref{thm:compet} is established through several intermediate
results:
\begin{enumerate}
\item A lower bound $L^{\opt}[s]$ is developed for the (difficult
  to evaluate) quantity $M^{\opt}[s]$, which yields a manageable upper
  bound to $\chi[s]$ (Prop.~\ref{bellb}). (An analog lower bound for
  the deterministic case can be found in \cite{PanagiotouS06}.)
\item An upper bound to the competitive ratio is evaluated for a quasi
  uniform distribution, which is analytically tractable
  (Prop.~\ref{prop:quasi-uniform-lopt}).
\item Finally, the analysis of the stochastic competitive ratio of an arbitrary
  distribution $s$ is reduced to that of a related, quasi uniform distribution
  $s'$ (Prop.~\ref{sigmaeta}).
\end{enumerate}
Steps 1 and 3 lead to the following chain of inequalities:
\begin{align}
  \chi[s]=\frac{M^{\lpr}[s]}{M^{\opt}[s]} \leq \frac{M^{\lpr}[s]}{L^{\opt}[s]}
  \leq \frac{M^{\lpr}[s']}{L^{\opt}[s']} \mf
\end{align}

\begin{proposition}\label{bellb}
  Under the LRUSM, the miss rate of OPT is bounded from below as $ M^{\opt}[s]
  \geq L^{\opt}[s]$, where, for $G\in\{1,\ldots,V-C\}$,
  \begin{align}\label{eq:lopt}
    &L^{\opt}_G[s]\deq G\left(\sum_{j=0}^{C+G-1} \frac{1}{1-S(j)}\right)^{-1}
    \mc &L^{\opt}[s] \deq \max_{G\in\{1,\ldots,V-C\}} L^{\opt}_G[s] \mf
  \end{align}
\end{proposition}

\begin{proof} For a given value of $G$, we consider a partition of the trace
  $a_1, a_2, \ldots$ into consecutive segments, each minimal under the
  constraint that it contains exactly $C+G$ distinct references. Let $\tau_i$
  denote the number of steps in the $i$-th such segment.  The random variables
  $\tau_1, \tau_2, \ldots$ are statistically independent and identically
  distributed. Any one of them, generically denoted $\tau$, can be decomposed as
  \begin{equation}
    \tau \deq \sum_{j=0}^{C+G-1} \phi_j \mc
  \end{equation}
  where $\phi_j$ is the minimum number of steps, starting at some fixed
  time, to observe the first access with stack distance greater than $j$. It
  can be easily seen that $\phi_j$ is a geometric random variable, with
  parameter $p_j=1-S(j)$, $\Pb[\phi_j=k]=(1-p_j)^{k-1}p_j$, expected
  value $1/p_j$, and finite variance. By linearity of expectation we have:
  \begin{equation}
    \E[\tau] =  \sum_{j=0}^{C+G-1} \frac{1}{1-S(j)} \mf
  \end{equation}
  Under any policy, including OPT, in each of the intervals of
  durations $\tau_1,\tau_2, \ldots $, there occur at least $G$ misses,
  since at most $C$ of the referenced items could be initially in the
  buffer. Therefore, we can write the following chain of relations:
  \begin{align}
    \begin{split}
      M^{OPT}[s] \geq& \lim_{q\rightarrow\infty}
     \E\left[\frac{q G}{\sum_{i=1}^{q}\tau_i}\right]
     = G \lim_{q\rightarrow\infty} \E\left[\frac{1}{\bar \tau_q}\right]
      = G \E \left[ \frac{1}{\lim_{q\rightarrow\infty} \bar \tau_q}\right]
      = \frac{G}{\E[\tau]} = L^{\opt}_G[s]\mc
    \end{split}
  \end{align}
  where $\bar \tau_q\deq\sum_{i=1}^q \frac{\tau_i}{q}$. The
  interchange between limit and expectation is justified because (i)
  by the law of large numbers, $\bar \tau_q$ converges in distribution
  to the delta peaked at $\E[\tau]$ and (ii) the function $1/x$ is
  continuous and bounded within the support of $\bar \tau_q$ (which
  equals $\{x\in\Z : x\geq C+G\}$ since, for each $i$, $\tau_i \geq
  C+G$) \cite{Loeve77}.
\end{proof}

\begin{proposition}\label{prop:quasi-uniform-lopt}
  Let $s'$ be the distribution defined as
  \begin{equation}
    s'(j) \deq
    \begin{cases}
      \sigma & j=1\\
      \eta & j\in \{2,\ldots,V-1\}\\
      \eta' & j=V 
    \end{cases}\mf
  \end{equation}
  Then the following lower bound for OPT miss rate holds:
  \begin{equation}
    M^{\opt}[s'] \geq  L^{\opt}[s'] \geq \frac{V-C}{2}\left[ 1+\frac{1}{\eta} 
      \ln\left( \frac{(V-2)\eta+\eta'}{\left(\frac{V-2-C}{2}\right)\eta + \eta'}
      \right) \right]^{-1} \mf
  \end{equation}
\end{proposition}
\begin{proof}
  Applying Prop.~\ref{bellb} we can write
  \begin{align}
    L^{\opt}_G[s']&=G\left[1+\sum_{j=1}^{C+G-1} \frac{1}{1-S(j)}\right]^{-1}
    = G\left[1+\sum_{j=1}^{C+G-1} \frac{1}{\eta'+(V-j-1)\eta}\right]^{-1} \\
    &= G\left[1+\sum_{k=V-C-G}^{V-2} \frac{1}{\eta'+k\eta}\right]^{-1}
    \geq G\left[1 + \frac{1}{\eta} \ln \left(
        \frac{(V-2)\eta+\eta'}{(V-C-G-1)\eta + \eta'} \right) \right]^{-1} \mf
  \end{align}
  By setting $G=\frac{V-C}{2}$ the thesis follows.
\end{proof}

\begin{lemma}\label{belb}
  If $\forall j \; S_1(j) \geq S_2(j)$ then $L^{\opt}[s_1] \leq L^{\opt}[s_2]$.
\end{lemma}
\begin{proof}  From the definition of $L^{\opt}$ we can see that it is decreasing with
  $\sum_{g=0}^{C+G-1} \frac{1}{1-S(g)}$, and hence decreasing with each $S(j)$.
\end{proof}

\begin{proposition}\label{sigmaeta}
  Let $s$ be an LRU stack access distribution, $C$ the buffer capacity and $\lpr$
  its associated optimal online policy. Consider $s'$ defined as follows:
  \begin{equation}
    s'(j)=
    \begin{cases}
      \sigma & j=1\\
      \eta & j\in \{2,\ldots,L+D\}\\
      \eta' & j=L+D+1\\
      0 & j>L+D+1
    \end{cases}\mc
  \end{equation}
  where $\eta \deq \bar s(K+1,L)$, $\sigma \deq S(K) - (K-1)\eta$, $D \deq
  \left\lceil \frac{1-S(L)}{\eta} \right\rceil - 1$, $\eta' \deq 1-S(L)-D\eta$
  (note that $0<\eta'\leq \eta$), then $s'$ is a valid distribution and
  $L^{\opt}[s'] \leq L^{\opt}[s]$ and $M^{\lpr}[s'] = M^{\lpr}[s]$.
\end{proposition}
\begin{proof}  We begin by observing (due to Lemmas~\ref{sub-lemma} and \ref{sub-lemma-2})
  $\bar s(2, K)>\bar s(K+1, L)>\bar s(L+1, V)$.  The transformation from $s$ to
  $s'$ can be intuitively described as follows.
  \begin{compactitem}
  \item We first ``flatten'' to $\eta=\bar s(K+1, L)$ the access distribution
    within the segment $[K+1, L]$ (in this operation no probability mass is
    moved outside the segment).
  \item We set $s(j)=\eta$ in the interval $[2,K]$. Since $\bar s(2, K)>\bar
    s(K+1, L)$ some probability mass is removed from $[2, K]$.
  \item We add to $s(1)$ the probability removed during the previous step.
  \item We redistribute the probability mass $1-S(L)$ at positions starting from
    $L+1$, assigning $\eta$ per position, possibly followed by a leftover
    $\eta'$. This is possible because $\bar s(K+1, L)>\bar s(L+1, V)$.
  \end{compactitem}
  After all these movements total probability mass is preserved and no position
  has a negative value, hence $s'$ is still a valid distribution. Because of
  Lemmas~\ref{sub-lemma} and \ref{sub-lemma-2} the transformation yields
  $\forall j \; S'(j) \geq S(j)$ and hence, because of Lemma~\ref{belb},
  $L^{\opt}[s'] \leq L^{\opt}[s]$. As for LPR miss rate we have
  \begin{equation}
    M^{\lpr}[s] = 1-S(L) + (L-C)\eta = (L+D-C)\eta + \eta' = M^{\lpr}[s'] \mf    
  \end{equation}
\end{proof}

\begin{proof}[{\bf Proof of Thm.~\ref{thm:compet}}]
  Because of Prop.~\ref{sigmaeta}, given a distribution $s$ we can obtain
  a quasi uniform distribution $s'$ (with a narrowed memory space 
  $W \deq L+D+1\leq V$) such that
  \begin{equation}\label{chitilde}
    \chi[s] \leq \frac{M^{\lpr}[s']}{L^{\opt}[s']}
    = \frac{\eta'+(W-1-C)\eta}{W-C}  \left[2+\frac{2}{\eta}\ln\left(
        \frac{(W-2)\eta+\eta'}{\left(\frac{W-2-C}{2}\right)\eta+\eta'}\right)\right]
    \deq \tilde\chi[s'] \mf
  \end{equation}
  To simplify the analysis of $\chi[s]$ we can assume $C \gg 1$ and $W-C>2$ (the
  complementary case is easy to deal with). Equation (\ref{chitilde}) reduces to
  \begin{equation}
    \label{eq:chi}
    \chi[s] \leq  O(1) + 2 \ln\left(\frac{2(W-1)}{W-2-C}\right) \in
    O\left(\ln C\right) \mf
  \end{equation}
    Finally, since $M^{\lpr}[s] \leq (W-C)\eta$ and $\eta \leq \frac{1}{W-2}$
  (since $s(1)+\eta'+\eta(W-2)=1$) we have $M^{\lpr}[s] \leq \frac{W-C}{W-2}$, implying
  \begin{equation}
    \chi[s] \in O\left(\ln\frac{1}{M^{\lpr}[s]}\right) \mc
  \end{equation}
  which, for $M^{\lpr}[s]>\frac{1}{C}$, provides a more descriptive bound than
  (\ref{eq:chi}).
\end{proof}

\section{Fast Simulation of the LPR Policy}
\label{sec:fast-simul-optim}

In experimental studies it is important to simulate eviction
policies on sets of benchmark traces.  From the stack distances, the
number of misses for all buffer capacities can be easily derived in
time $O(V)$.  Previous work \cite{BennettK75,AlmasiCP02} has shown how
to compute the stack distance for the LRU policy in time $O(\log V)$
per access. We derive an analogous result for the more complex LPR
policy.

\begin{theorem}\label{thm:lpr-stack-up}
  Given any stack-distance distribution $s$, the number of misses
  incurred by the corresponding optimal LPR policy on an arbitrary
  trace of $N$ references, can be computed, simultaneously for all
  capacities, in time $O(V+N \log V)$.
\end{theorem}

The algorithm proposed to prove the preceding result exploits some
relations between the LPR stack and the LRU stack and achieves
efficiency by means of fast data structures.

\begin{proposition}\label{thm:lru-lpr}
  Let $s$ be a stack-distance distribution and let $(q_1, q_2, \ldots,
  q_l)$ be the (increasing) sequence associated with $s$ as in
  Theorem~\ref{thm:aveCap}. Let $\Lambda$ be the LRU stack and let
  $\Pi$ be the LPR stack corresponding to $s$, both assumed initially
  equal. In either stack, let the $i$-th segment be the set of
  positions in interval $Q_i\deq[q_i+1, q_{i+1}]$. Finally, for $j \in
  \{1,2, \ldots, V\}$, let $\rho_t(j)$ denote the position in the LPR
  stack $\Pi_t$ of the item that is in position $j$ of the LRU stack
  $\Lambda_t$, i.e., $\Pi_t(\rho_t(j))$=$\Lambda_t(j)$. Then,
  we have:

  \begin{enumerate}
   
  \item {\sc Segment equivalence.} At any time $t$, segment $Q_i$
  contains the same items in both stacks, that is, $j \in Q_i$ if and
  only if $\rho_t(j) \in Q_i$.

  \item {\sc Relative update.} Upon an access $a_t$ at LRU stack
  depth $d_t$, if $d_t=1$, then the map $\rho$ between the two
  stacks is unchanged. Otherwise ($d_t>1$), $\rho$ is updated as
  follows. For any $t$, any segment $Q_i$, with $i=1,\ldots, l-1$, and
  any $h \in \{0, q_{i+1}-q_i-1\}$, (so that $(q_i+1+h) \in Q_i$), we
  have:
\begin{align}
      \rho_{t+1}(q_i+1+h)= \begin{cases} \rho_{t}(q_i+1+h) & d_t < q_i+1+h\\ 
   \rho_{t}\left(q_i+1 + \left((h-1) \bmod (d_t-q_i) \right) \right) 
      & q_i +1 +h \leq d_t \leq q_{i+1} \\
   \rho_{t}\left(q_i +1 + \left( (h-1) \bmod (q_{i+1}-q_i) \right) \right) 
      & q_{i+1} < d_t \mf \\  
  \end{cases}
\end{align} 

  In other words, (a) below the LRU point of access, $\rho$ does not
  change; (b) within the prefix of the segment including the access
  up to the point of access itself, as well as (c) within those
  segments that are entirely above the point of access, $\rho$ incurs
  a unit, right cyclic shift.  \end{enumerate}
\end{proposition}
  
\begin{proof}
{\sc Segment equivalence}. We begin by observing that, for buffer sizes
in the set $\{q_1, q_2, \ldots, q_l \}$, the LPR policy coincides with
LRU, whence $B_t^{LPR}(q_i) =B_t^{LRU}(q_i)$. In fact, for a $\kl$
policy, the buffer content is a subset of the first $L$ positions of
the LRU stack. From Corollary~\ref{thm-KL-2-comp}, if $C=q_i$, then
$L(C)=q_i$, hence the LPR buffer content must equal that of the first
$q_i$ positions of the LRU stack which, by definition of stack, is
also the content of the LRU buffer of capacity $q_i$. The segment
equivalence property follows since, for any stack policy, the content
of the stack in a segment $Q_i$ equals (by definition) the set
theoretic difference between $B(q_{i+1})$ and $B(q_i)$.\\

\noindent{\sc Relative update}. If $d_t=1$, then neither stack changes,
therefore $\rho_{t+1}=\rho_t$. Otherwise, let $d_t \in Q_a$, that is,
let $Q_a$ be the segment capturing the access and consider the
following cases.

\begin{romenum}
\item For all buffer sizes $C>q_{a+1}$, the access is a hit both for the LRU and
  for the LPR policy, hence both stacks, and consequently the $\rho$ map, remain
  unchanged at positions greater than $q_{a+1}$. This is a subcase of case (a)
  in the statement and applies to all segments with $i>a$ (hence $1+q_i+h
  >q_{a+1}$).

\item For buffer sizes $C \in Q_a$ the situation is as follows.  Under LRU,
  there is a miss for $C<d_t$, so that in the LRU stack: the items in positions
  smaller than $d_t$ shift down by one position, the item at $d_t$ goes at the
  top of the stack, and all other items retain their position. Under LPR, let
  $C_a$ be the smallest capacity of a buffer containing the referenced item,
  $a_{t+1}$. Then, all the buffers with $C<C_a$ will evict item
  $\Lambda_{t}(q_a)$, which will go to position $C_a$ of the LPR stack, the
  item previously at $C_a$ will go to the top of the stack, while all remaining
  items in segment $Q_a$ will retain their position. Consequently, for
  $h=(d_t+1)-(q_a+1), \ldots, q_{a+1}-(q_a+1)$, we have
  $\rho_{t+1}(q_a+1+h)=\rho_{t}(q_a+1+h)$, still a subcase of case (a) in the
  statement. Furthermore, $\rho_{t+1}(q_a+1)=\rho_t(d_t)$ and, for $h=1, \ldots,
  d_t-(q_a+1)$, $\rho_{t+1}(q_a+1+h)=\rho_{t}(q_a+1+h-1)$, which establishes
  case (b) in the statement.

\item For segments $Q_i$ with $i<a$, the argument is a straightforward
  adaptation of that developed for case (ii) and establishes case (c) of the
  statement.
\end{romenum}
\end{proof}

\comment{**************************************************************
\begin{figure}
  \centering
  \includegraphics[width=6cm]{images/stack-update}
  \caption{LPR stack update (steady state). $b_t$ is the LPR stack depth
   of the access $a_t$ and $y_t(C)$ the eviction made by the buffer of
   capacity $C$.}
  \label{fig:stack-up}
\end{figure}
end-of-comment*********************************************************}

We are now ready to provide the algorithm for LPR stack-distance
computation and its analysis.

\begin{proof}[{\bf Proof of Thm.~\ref{thm:lpr-stack-up}}] The work of
\cite{BennettK75,AlmasiCP02} has provided a procedure that, given the 
initial LRU stack $\Lambda_0$ and the prefix $a_1, \dots, a_t$ of the
input trace, will output the LRU stack distances $d_0,\ldots,d_{t-1}$, in
time $O(V+ t \log V)$. Below, we develop a representation of the map
$\rho$ between the LRU and the LPR stack, which can be updated and
queried in time $O(\log V)$ per access. Then the LPR stack distance of
access $a_{t+1}$ can be obtained as $\rho(d_t)$.

By Proposition~\ref{thm:lru-lpr}, we can represent $\rho$ by a
separate sequence $(\rho(q_i+1), \ldots, \rho(q_{i+1}))$ for each
segment. On such a sequence, we need to perform cyclic shifts of an
arbitrary prefix and to access element $\rho(q_i+1+h)$, given an $h
\in [0,q_{i+1}-(q_i+1)]$. Any of the well-known dynamic balanced trees
(AVL, 2-3, red-black, \dots) \cite{CormenLRS04} can be easily adapted to
perform each of the required operations in time
$O(1+\log(q_{i+1}-q_i)) = O(\log V)$. Hereafter, we denote by $R_i$
the data structure for segment $Q_i$.

The number of segments where the map $\rho$ can change in one step can
be $\theta(V)$, in the worst case. Therefore, we adopt a \emph{lazy}
update strategy whereby only the segment $Q_a$ capturing the access is
actually updated; for the other segments, record is taken that a shift
should be applied to the sequence, without performing the shift
itself. It is sufficient to increment a counter storing the amount of
shift that has to be applied to the sequence and then perform just one
global rotation when the segment is accessed. Still, individually
incrementing each segment counter could lead to work proportional to
$V$, per step. Instead, we maintain an
\emph{auxiliary tree} $T$ of counters which collectively serve the 
segments and where at most logarithmically many counters need updating
in a given step. A further field in the auxiliary tree will enable
quick identification of segment $Q_a$.

More specifically, the auxiliary tree is a static, balanced tree with
$l-1$ leaves corresponding, from left to right, to segments $Q_1,
\ldots Q_{l-1}$. In each internal node, a search field contains the
maximum right boundary of any segment associated with a descendant of
that node. In each node, a counter field will be maintained so that,
at the end of a step, the sum of the counters over the ancestors of
the $i$-th leaf represent the amount of shift to be applied to
$\rho$-sequence in segment $Q_i$.

With the above data structures in place, the algorithm to process one
access $a_{t+1}$ is outlined next.
\begin{compactenum}
\item From the LRU procedure, obtain the LRU stack distance $d_t$.
\item In the auxiliary tree $T$, with the help of the search field, 
traverse the path $\cal P$ from the root to the leaf corresponding to
the segment $Q_a$ which contains $d_t$.
\item While traversing $\cal P$, increment by one the counter of a
left child $\nu$ of a visited node whenever $\nu$ itself is not on
${\cal P}$. (This operation corresponds to incrementing the shift
count for all the segments to the left of $Q_a$, as required by
Proposition~\ref{thm:lru-lpr}.)
\item While traversing $\cal P$, add the counters of the visited 
nodes and apply a shift of the resulting amount to $R_a$. Subtract such
amount from the counter of the leaf for $Q_a$.
\item Read $\rho_t(d_t)$ from the $(d_t-q_a)$-th position of sequence 
$R_a$ and output this value as the LPR stack distance of $a_{t+1}$.
\item Apply a unit right cyclic shift to the prefix of length $(d_t-q_a)$
of sequence $R_a$.
\end{compactenum}
Each step in the outlined procedure can be accomplished in time
$O(\log V)$, so that the overall time for processing $N$ accesses is
$O(V+N \log V)$, where the term $V$ accounts for the initial set up of
the data structures.

To avoid that the counter in a node $\nu$ of the auxiliary tree grow
unbounded, we observe that it is sufficient to maintain its value
modulo the minimum common multiple of the lengths of the segments
associated with the descendant leaves of $\nu$.
\end{proof}

\section{On Finite Horizon}
\label{sec:fh}

In practice, when dealing with sufficiently long traces, a policy that
is optimal over an infinite horizon is likely to achieve near optimal
performance.  For shorter traces, transient effects may play a
significant role, whence the interest in optimal policies over a
finite horizon.  In principle, the optimal policy can be computed by a
dynamic-programming algorithm based on (\ref{beqsys1}), but the
exponential number of states makes this approach of rather limited
applicability.  An alternate, often successful route consists in
guessing a closed form characterization of a policy $\pi$ and its
corresponding optimal cost function $J^{\pi}_{\tau}(\cdot)$.  Under
very mild conditions, if the guess satisfies (\ref{beqsys1}), then
$\pi$ is an optimal policy. Unfortunately, we have been unable to find
a tractable form for the optimal cost. Ultimately, we have
circumvented this obstacle for monotone stack-depth distributions, by
realizing that what is really needed to make an optimal choice between
two states is not the absolute value of their costs, but rather their
relative value.
\begin{theorem}\label{lruthm}
  Let $s$ be non increasing, i.e., $s(j)\geq s(j+1)$ for $ j\in\{1,
  V-1\}$.  Then, for any finite horizon $\tau \geq 1$ and any initial
  buffer content, LRU is an optimal eviction policy.
\end{theorem}
\begin{theorem}\label{mruthm}
  Let $s$ be non decreasing, i.e., $s(j)\leq s(j+1)$ for $ j\in\{1,
  V-1\}$.  Then, for any finite horizon $\tau \geq 1$ and any initial
  buffer content, MRU is an optimal eviction policy.
\end{theorem}
Thus, for monotone stack-depth distributions, the finite horizon
optimal policy is time invariant, hence it is also optimal over an
infinite horizon.  This property does not hold for arbitrary
distributions.  In spite of the symmetry between the above two
theorems, their proofs, given in \S\ref{sec:proof-thm.-reflruthm} and
\S\ref{sec:proof-thm.-refmruthm}, require significantly different
ideas.

We have extended Thm.~\ref{lruthm} to the case of (non increasing)
\emph{dependent} stack depth distribution, where, given a prefix trace
$\zeta$, the next stack distance is described by the following
distribution, assumed to be non increasing for all $\zeta$:
\begin{align}
  s_\zeta(i) = \Pb\left[d_t = i | \zeta \right] \mf
\end{align}
The details are given in \S\ref{sec:gener-depend-proc-lru},
Thm.~\ref{thm:gener-depend}. For the case of initially empty buffer a
result similar to Thm.~\ref{thm:gener-depend}, but based on a stronger
notion of optimality, is derived by Hiller and Vredeveld
\cite{HillerV09}, using different techniques. This generalization of
the LRUSM has also been studied by Becchetti \cite{Becchetti04}, who
provides sufficient conditions on $s_\zeta$ and $C/V$ for the
stochastic competitive ratio of LRU against OPT to be $O(1)$.

\subsection{Non-Increasing Access Distribution}
\label{sec:proof-thm.-reflruthm}

For the purposes of this section, the state description for the LRUSM
developed in \S\ref{sec:lru-stack-model} can be simplified, by
unifying the representation of the LRU stack and of the buffer in a
vector $x$ of $V$ binary components, where $x_t(j)=1$ when the item at
depth $j$ in the LRU stack is in the buffer at time $t$ and $x_t(j)=0$
otherwise.  The disturbance is still the LRU depth of the access
($w_t=d_t$), while the control $u_t$ specifies the LRU depth of the
item to be evicted.  Denoting by $f$ the state transition function, we
have:
\begin{equation}
  x_{t+1}=f(x_t,d_t,u_t) \mf  
\end{equation}
With this representation, the well-known LRU policy amounts to
evicting the item in the deepest position of the (resulting) stack,
among those that are in the buffer:
\begin{definition}
  Let $R_d(x)$ denote the state resulting by applying a unit right
  cyclic shift to the prefix of length $d$ of $x$; (strictly speaking,
  if $x(d)=0$ then $R_d(x)$ is a pseudo-state, as it is not in the
  admissible state set).  The \emph{Least Recently Used} (LRU) policy
  is defined (for a miss, $x(d)=0$) by
  \begin{equation}\label{eqn:lru}
      \lru(x,d) = \max\{j:~y(j)=1, \;\text{ where } y=R_d(x)\}\mf
  \end{equation}
\end{definition}

\begin{definition}
  We say that two states $y$ and $z$ are form a \emph{critical pair}
  and write $y <_c z$ if their structure is related as follows, where
  $\nu,\iota$ are arbitrary and $\sigma \in 0^*$:
\begin{equation} \begin{split} y &=
  1\nu1\iota0\sigma \mc\\ z &= 1\nu0\iota1\sigma \mf \end{split}
  \end{equation}
We also write $y \leq_c z$ when $y = z$ or $y <_c z$.
\end{definition}

\begin{remark}
A critical pair represents a choice between what would the LRU policy
do (obtaining $y$) and what would a different eviction policy do
(obtaining $z$), when choosing the item to evict after the stack
rotation.
\end{remark}
\begin{lemma}\label{lem:critp}
  The evolution of a critical pair under LRU preserves its criticality
  and order:
  \begin{equation}
    \forall y \, \forall z \, :\,  y <_c z \quad \forall d \quad 
    y' = f^{\lru}(y,d) \leq_c z' f^{\lru}(z,d) \mf
  \end{equation}
where $f^{\lru}(x,d)=f(x,d,\lru(x,d))$.
\end{lemma}
\begin{proof}
    We analyze the four possible cases:
  \begin{compactitem}
  \item Hit for both $y$ and $z$. The two stacks rotate and produce a critical
    pair with $y' <_c z'$.
  \item Miss for both $y$ and $z$. The two evictions in the last filled
    positions make the states equal if $\iota \in 0^*$, otherwise they yield $y'
    <_c z'$.
  \item Hit for $y$ and miss for $z$. The eviction in $z$ yields $y'=z'$.
  \item Miss for $y$ and hit for $z$. The eviction in $y$ brings $y'=z'$ if
    $\iota \in 0^*$ and $y' <_c z'$ otherwise.
  \end{compactitem}
\end{proof}

\begin{proposition}\label{lruthm2}
  Let $s$ be monotonic non decreasing, then LRU is the optimal eviction policy
  for every time horizon $\tau$ and any initial buffer content. Furthermore:
  \begin{equation}
    \forall \tau \; \forall y \, \forall z \, :\, y \leq_c z \quad
    J^*_\tau(y)\leq J^*_\tau(z) \mf
  \end{equation}
\end{proposition}
\begin{proof}[Proof (by induction on $\tau$)]
  \par{\bf Base case.} For $\tau=1$ we have
  \begin{equation}
    \forall x \quad J^*_1(x)=\E_d\left[g(x,d)\right] \deq \bar g (x) \mc
  \end{equation}
  which, using the monotonicity of $s$, yields
  \begin{equation}
    \forall y \, \forall z \, :\, y \leq_c z \quad J^*_1(y)\leq
    J^*_1(z) \mf
  \end{equation}
\noindent \par{\bf Induction.} Assuming now that the statement holds for all 
$t < \tau$, we obtain
  \begin{equation}
    \begin{split}
      \forall x \quad J^*_\tau(x)&=\bar g(x) + \E_d\left[\min_u
        J^*_{\tau-1}\left(f(x,d,u)\right)\right]\\
      &=\bar g(x) + \E_d\left[J^*_{\tau-1}\left(f^{\lru}(x,d)\right)\right]\mf
    \end{split}
  \end{equation}
  Since $\bar g(y)\leq \bar g(z)$ and since, by the inductive hypothesis and
  Lemma~\ref{lem:critp},
  \begin{equation}
    J^*_{\tau-1}\left(f^{\lru}(y,d)\right)
    \leq J^*_{\tau-1}\left(f^{\lru}(z,d)\right) \mc
  \end{equation}
  we finally obtain $J^*_\tau(y) \leq J^*_\tau(z)$.
\end{proof}
\begin{proof}[{\bf Proof of Thm.~\ref{lruthm}}]
  Thm.~\ref{lruthm} follows directly from Prop.~\ref{lruthm2}.
\end{proof}

\subsubsection{Generalization to Dependent Processes}
\label{sec:gener-depend-proc-lru}

Let $\mathcal{V}\deq\{1, 2, \ldots, V\}$,
$\mathcal{V}^0\deq\{\epsilon\}$ ($\epsilon$ being the null trace),
$\mathcal{V}_L =\cup_{i=0}^L \mathcal{V}^i$ be the set of traces of
size no greater than $L$, and $X$ the state space (boolean vectors on
the LRU stack).  We consider stochastic processes generating a trace
of length $L$, such that, after having generated a partial trace
$\zeta$ of length $t-1=|\zeta|$, the probability distribution of the
next access is specified by
\begin{equation}\label{eq:stat-w}
 s_\zeta(i)= \Pb\left[a_{t+1}=\Lambda_t(i)\right|\zeta]
           =\Pb\left[d_t=i\right|\zeta] \mf
\end{equation}
The optimal cost achievable for such a process, given a partial 
trace $\zeta$ is $\forall x \in X$,
\begin{align}
  J^*_L(\zeta,x) = \E_d\left[\min_u\left\{g(x,d)+J^*_L\left(\zeta
        d,f(x,d,u)\right)\right\} |\zeta \right]
\end{align}
where the probability distribution of $d$ is a function of $\zeta$ 
as given in (\ref{eq:stat-w}).

\begin{theorem}\label{thm:gener-depend}
  If $\forall \zeta$ $s_\zeta$ is a non-increasing function, then LRU
  is optimal for any initial buffer content:
 \begin{align} 
   \forall L \; \forall \zeta \in \mathcal{V}_L \; \forall y,z \in X : 
   y \leq_c z \quad J_L^*(\zeta,y) \leq J_L^*(\zeta,z) \mf
 \end{align}
\end{theorem}
\begin{proof}
  Let $L$ be given. For a trace $\zeta$ we define $ \tau_\zeta \deq L -
  |\zeta|$. The proof is by induction on $\tau_\zeta$.
  \par{\bf Base case.} 
  \begin{align}
    & \forall \zeta : \tau_\zeta=1 \; \forall x\in X \quad
    J^*_L(\zeta,x)=\E_d\left[g(x,d)\right |\zeta]\\
    \Rightarrow \quad & \forall y \leq_c z \quad J^*_L(\zeta,y) \leq 
    J^*_L(\zeta,z) \mf
  \end{align}
  \par{\bf Induction.}
  Since by inductive hypothesis we assume that $\forall \theta :
  \tau_\theta < \tau$
  \begin{equation}
    \forall y \leq_c z \quad J^*_L(\theta,y) \leq J^*_L(\theta,z) \mc
  \end{equation}
  we obtain, $\forall x \in X$,
  \begin{align}
    J^*_L(\zeta,x) &= \E_d\left[\min_u\left\{g(x,d)+J^*_L\left(\zeta
          d,f(x,d,u)\right)\right\}|\zeta\right]\\
    &= \E_d\left[g(x,d)|\zeta\right] + E_d\left[J^*_L\left(\zeta d,
        f^{\lru}(x,d)\right)|\zeta\right] \mf
  \end{align}
  Let $\theta \deq \zeta d$, since $\forall y \leq_c z \quad
  J^*_L(\theta,y) \leq J^*_L(\theta,z)$ and, by the inductive
  hypothesis and Lemma~\ref{lem:critp},
  \begin{equation}
    J^*_L\left(\theta, f^{\lru}(y,d)\right)
    \leq J^*_L\left(\theta, f^{\lru}(z,d)\right) \mc
  \end{equation}
  we finally obtain $J^*_L(\zeta, y) \leq J^*_L(\zeta, z)$.
\end{proof}

\subsection{Non-Decreasing Access Distribution}
\label{sec:proof-thm.-refmruthm}
In this subsection we prove that MRU is the optimal eviction policy for
non-decreasing $s$ for any time horizon. The proof will be by induction: by
assuming the optimal policy to be MRU for $t\leq \tau$ we will be able to prove
its optimality for the time horizon $\tau+1$ (more precisely, a strengthened
inductive hypothesis will be used).

In order to compare costs under MRU for different initial states we introduce a
useful partition of the misses.  Imagine to place an observer on every
out-of-buffer item, following the item going down the LRU stack during the
system evolution; every time an out-of-buffer item is accessed its observer
moves to the item evicted by the policy $\mu$. Thus the set $\Psi$ of the
observers remains constant during the evolution.

Let $d_t$ be the access depth at time $t$, let $\psi$ be an observer
and $l_t^\mu\left(\psi,x_0,d_{t'<t}\right)$ its LRU stack depth at
time $t$. We are interested in the event \emph{the item observed by
$\psi$ is accessed at time $t$}:
$d_t=l_t^\mu\left(\psi,x_0,d_{t'<t}\right)$.  We can partition the
misses occurring in $\tau$ steps attributing each miss to the observer
$\psi\in\Psi$ on the item currently accessed:
\begin{equation}
  J_{\tau}^\mu(x_0)=\sum_{t=0}^{\tau-1}\Pb_{\miss}(t)=\sum_{t=0}^{\tau-1}
  \sum_{\psi\in\Psi}\Pb[d_t=l_t^\mu(\psi,x_0,d_{t'<t})]=
  \sum_{\psi\in\Psi}\sum_{t=0}^{\tau-1}\Pb[d_t=l_t^\mu(\psi,x_0,d_{t'<t})] \mf
\end{equation}

Let $\psi_j$ be the observer which is at depth $j$ at time 0.  Under
MRU the evolution of an observer $\psi_j$ does not depend on $x_0$ but
only on the initial position of the observed item (i.e.,
$l_t^{\mru}(\psi_j,x_0,d_{t'<t}) =
l_t^{\mru}(\psi_j,d_{t'<t})=l_t(\psi_j)$ for brevity).
If we have two states $x'$ and $x''$ which differ for only two observers
$\psi_i$ and $\psi_j$ we can write their costs $\Gamma'$ and $\Gamma''$ as:
\begin{equation}\begin{split}
    \Gamma'&=\hspace{-1.3em}\sum_{\psi\in\Psi\setminus\{\psi_i\}}
    \sum_{t=0}^{\tau-1}\Pb[d_t=l_t(\psi)]+
    \sum_{t=0}^{\tau-1}\Pb[d_t=l_t(\psi_i)] \mc\\
    \Gamma''&=\hspace{-1.3em}\sum_{\psi\in\Psi\setminus\{\psi_j\}}
    \sum_{t=0}^{\tau-1}\Pb[d_t=l_t(\psi)]+
    \sum_{t=0}^{\tau-1}\Pb[d_t=l_t(\psi_j)] \mf\\
  \end{split}\end{equation}
where the first term is equal in both the costs, because it is due to observers
which start in the same position for both states, and thus:
\begin{equation}\begin{split}
  \Gamma'-\Gamma''=
  \sum_{t=0}^{\tau-1}\Pb[d_t=l_t(\psi_i)]-\sum_{t=0}^{\tau-1}\Pb[d_t=l_t(\psi_j)]=
  \gamma_\tau(i)-\gamma_\tau(j) \mc
  \end{split}\end{equation}
having defined $\gamma_\tau(i)\deq\sum_{t=0}^{\tau-1}\Pb[d_t=l_t(\psi_i)]$. Thus,
the difference in the costs depends only on the items observed by the
different observers.
Quantity $\gamma_\tau(i)$ represents the contribution to the total cost due to
items observed by \(\psi_i\), the observer that at time zero is in position $i$
(not in the buffer).

To prove that MRU is optimal for a time horizon of $\tau+1$ under the hypothesis
that it is optimal for any $t\leq \tau$ it is sufficient to prove that
\(\gamma_\tau(i)\leq\gamma_\tau(j)\) if \(i<j\):
\begin{proposition}\label{sec:proof-thm.-refmruthm-1}
  Let $s$ be non decreasing: $\forall j\in\{1, V-1\} \;\; s(j)\leq s(j+1)$.  If
  \begin{equation}\begin{split}
      \forall t\leq\tau \quad \forall i \, \forall j : i<j \qquad
      \gamma_t(2)\leq\gamma_t(i)\leq\gamma_t(j)\leq 1+\gamma_t(2)  \mc
    \end{split}\end{equation}
  then
  \begin{equation}\begin{split}
      \forall i\,\forall j : i<j \qquad
      \gamma_{\tau+1}(2)\leq\gamma_{\tau+1}(i)\leq\gamma_{\tau+1}(j)\leq 1+\gamma_{\tau+1}(2) \mf
    \end{split}\end{equation}
\end{proposition}
\begin{proof}
  \par{\bf Base case.}  For $t=1$ we have $\forall k \;\; \gamma_t(k)=s(k)$,
  and hence
  \begin{equation}
    \gamma_t(2)\leq\gamma_t(i)\leq\gamma_t(j) \mf
  \end{equation}
  Furthermore we have that
  \begin{equation}
    \gamma_t(k)=s(k)\leq 1\leq 1+\gamma_t(2) \mf
  \end{equation}
  \par{\bf Induction.}
  \begin{align}
    \begin{split}
      \gamma_{\tau+1}(i)&=s(i)(1+\gamma_{\tau}(2))+S(i-1)\gamma_{\tau}(i)+(1-S(i))\gamma_{\tau}(i+1)\\
      &\leq s(i)(1+\gamma_{\tau}(2))+(1-s(i))\gamma_{\tau}(i+1)\\
      &\leq s(i)(1+\gamma_{\tau}(2))+(1-s(i))\gamma_{\tau}(j)\\
      &=s(i)(1+\gamma_{\tau}(2))+s(j)\gamma_{\tau}(j)+(1-s(i)-s(j))\gamma_{\tau}(j) \mc
    \end{split}
  \end{align}
  \begin{align}
    \begin{split}
      \gamma_{\tau+1}(j)&=s(j)(1+\gamma_{\tau}(2))+S(j-1)\gamma_{\tau}(j)+(1-S(j))\gamma_{\tau}(j+1)\\
      &\geq s(j)(1+\gamma_{\tau}(2))+(1-s(j))\gamma_{\tau}(j)\\
      &=s(j)(1+\gamma_{\tau}(2))+s(i)\gamma_{\tau}(j)+(1-s(i)-s(j))\gamma_{\tau}(j)\\[5mm]
      &\Rightarrow \gamma_{\tau+1}(i)\leq\gamma_{\tau+1}(j) \mf
    \end{split}
  \end{align}
  Finally under MRU we have
  \begin{align}
    \begin{split}
      \gamma_{\tau+1}(k)&=s(k)(1+\gamma_{\tau}(2))
      +S(k-1)\gamma_{\tau}(k)+(1-S(k))\gamma_{\tau}(k+1)\\
      &\geq \min\left\{1+\gamma_{\tau}(2),\gamma_{\tau}(k),
        \gamma_{\tau}(k+1)\right\}\\&=\gamma_{\tau}(k) \mc
    \end{split}
  \end{align}
  and
  \begin{align}
    \begin{split}
      \gamma_{\tau+1}(k)&=s(k)(1+\gamma_{\tau}(2))+S(k-1)\gamma_{\tau}(k)+(1-S(k))\gamma_{\tau}(k+1)\\
      &\leq \max\left\{1+\gamma_{\tau}(2),\gamma_{\tau}(k),
        \gamma_{\tau}(k+1)\right\}\\&=1+\gamma_{\tau}(2) \mc
    \end{split}
  \end{align}
  and therefore
  \begin{equation}
    \gamma_{\tau+1}(k)\leq 1+\gamma_{\tau+1}(2) \mf
  \end{equation}
\end{proof}
\begin{proof}[{\bf Proof of Thm.~\ref{mruthm}}]
  Thm.~\ref{mruthm} follows directly from 
  Prop.~\ref{sec:proof-thm.-refmruthm-1}.
\end{proof}

\section{On Bias Optimality}
\label{sec:ih}

Bias optimality is a stronger property than average (gain) optimality, since it
also takes into account the cost minimization in transient states of the
dynamical system (whereas average costs are insensible to policy changes in
transient states, provided the set of recurrent states stays unchanged).
Bias optimal policies are characterized as solutions of the Bellman equation
(see Prop.~\ref{bel}). In this section we provide evidence of the hardness of
the general solution of the Bellman equation in two ways:
\begin{itemize}
\item We prove that bias-optimal policies in general do not satisfy the
  inclusion property (the same result also applies to optimal policies in a
  finite horizon).
\item We derive the complex solution of the Bellman equation for the relatively
  simple case of $C=2$.
\end{itemize}

\begin{theorem}\label{nonbias}
  There are systems for which the unique optimal policy over some finite horizon
  and bias-optimal over infinite horizon is not a stack policy.
    \end{theorem}
\begin{proof}  We will exhibit a counterexample of a distribution $s$ that has optimal
  policies not induced by a priority (and hence not a stack policy). In more
  detail we first obtain by dynamic programming (executed by a computer program)
  the finite horizon optimal policies for two different buffer capacities $C'$
  and $C''$ (being $C' < C''$). Starting with buffers that satisfy the inclusion
  ($B_0(C') \subseteq B_0(C'')$) we show that there exists a temporal horizon $\tau$
  and state positions $j'$ and $j''$ such that, when in a state with both
  positions filled, for $C=C'$ the optimal policy evicts at $j'$, whereas for
  $C=C'$ it evicts at $j''$.
    By solving (by a computer program) the Bellman equation associated to the
  system we also prove that a similar situation applies to the infinite horizon
  case, implying that the unique bias-optimal policy in infinite horizon does
  not have the inclusion property.

  Consider the following $s$ distribution, with $V=8$ and $\beta=\frac{1}{16}$.
  \begin{center}
    \begin{tabular}{r|c|c|c|c|c|c|c|c|}
                                                \multicolumn{1}{r|}{$s(j)$:}
      &$\beta$&$3\beta$&$3\beta$&$0$&$4\beta$&$0$&$0$&$5\beta$\\
      \multicolumn{1}{c}{}\\[-3mm]
      \multicolumn{1}{c}{$j\in[1,V]$}& \multicolumn{1}{c}{1}&
      \multicolumn{1}{c}{2}& \multicolumn{1}{c}{3}& \multicolumn{1}{c}{4}&
      \multicolumn{1}{c}{5}& \multicolumn{1}{c}{6}& \multicolumn{1}{c}{7}&
      \multicolumn{1}{c}{8}\\[3mm] \multicolumn{9}{c}{}\\
    \end{tabular}
  \end{center}
  We are given an initial LRU stack $\Lambda_0$ and we consider the following
  initial buffers $B_0(2)$ and $B_0(3)$, satisfying the inclusion property $B_0(2)
  \subset B_0(3)$:
  \begin{align}
    B_0(2)&=\left[\Lambda_0(1), \Lambda_0(4) \right] \mc &
    B_0(3)&=\left[\Lambda_0(1), \Lambda_0(4), \Lambda_0(7) \right] \mf
  \end{align}
  If an access arrives at $x_0=\Lambda_0(8)$ a miss occurs in both buffers, and
  hence an eviction is needed. By computing the optimal policy for a time
  horizon of $T=5$ we see that
  \begin{compactitem}
  \item for $C=2$ the (unique) optimal eviction is at depth 2 ($\Lambda_0(1)$),
  \item for $C=3$ the (unique) optimal eviction is at depth 5 ($\Lambda_0(4)$).
  \end{compactitem}
  After the optimal evictions the two buffers become
  \begin{align}
    B_1(2)&=\left[\Lambda_0(8), \Lambda_0(4) \right] = \left[\Lambda_1(1),
      \Lambda_1(5) \right] \mc \\
    B_1(3)&=\left[\Lambda_0(8), \Lambda_0(1), \Lambda_0(7) \right] =
    \left[\Lambda_1(1), \Lambda_1(2), \Lambda_1(8) \right] \mc
  \end{align}
  hence violating the inclusion property.  The same eviction choice is given by
  the solution of the Bellman equation in infinite horizon, proving that
  bias-optimal policies are not, in general, stack policies.
    (Intuitively, the inclusion property violations happens because having in the
  buffer $\Lambda(8)$ decreases the ``profit'' of having $\Lambda(5)$: in fact a
  possible subsequent access to $\Lambda(8)$ brings to a new state with high
  instant cost $g$, since it has in the buffer the low profit item $\Lambda(6)$;
  whereas when $C=2$ the same access causes a miss, thus enabling the eviction
  of the poorly profitable $\Lambda(6)$.)
                  \end{proof}

\subsection{The Bias-Optimal Policy for $C=2$}
\label{sec:bias-optimal-C2}

When $C=2$, the state of our dynamical system can be identified by the
unique index $j\in\{2, \ldots, V\}$ such that the buffer contains the
items in positions $1$ and $j$ of the LRU stack.  The Bellman equation
becomes
\begin{equation}\label{bellman}
  h(j)= 1 - s(j) + h(2) - \lambda
  +S(j-1) \min\left\{0, h(j)-h(2)\right\}
  +(1-S(j)) \min\left\{0, h(j+1)-h(2)\right\}\mf
\end{equation}
The $h(j)$'s are defined up to an additive constant, so we can set
$h(2)=0$ to simplify the equation:
\begin{equation}\label{bellman-ext}
  h(j)=1 - s(j) - \lambda
  +S(j-1) \min\left\{0, h(j)\right\}
  +(1-S(j)) \min\left\{0, h(j+1)\right\}\mf
\end{equation}
The solutions will satisfy
  $h(j)=h(j)-h(2)=\lim_{\tau\rightarrow +\infty} J_\tau^*(j)-J_\tau^*(2)$.
We now ``guess'' the rather complex form of the solutions, in terms of the
auxiliary functions

\begin{equation}\begin{split}
    \beta(j)& \deq\max_{l\geq 1} \bar s(j, j+l-1) \mc\\
    \Phi(j) &\deq \Big\{l\geq 1: \forall k\in\left\{0, \dots, l-1\right\}
    \bar s(j+k, j+l-1)\geq \beta(2)\Big\} \cup \big\{0\big\} \mc\\
    \phi(j) &\deq \max \, \Phi(j)  \mc \hspace{1cm}
        \rho(j) \deq
    \begin{cases}
      \bar s(j, j+\phi(j)-1)-\beta(2) & \phi(j)\not = 0 \mc\\
      \beta(j)-\beta(2) & \phi(j) = 0 \mf\\
    \end{cases}
  \end{split}\end{equation}
where $\Phi(j)$ is the subsequence of items that are visited when applying the
policy induced using $\beta$ as a priority and $\phi(j)$ the length of this
subsequence.

\begin{proposition}\label{m2bel}
  Bellman equation (\ref{bellman-ext}) is solved using the following $\lambda$
  and $h(j)$:
  \begin{equation}
    \begin{split}
      \lambda&=1-\beta(2) \mc\\
      h(j)&=\beta(2)-s(j)-\frac{S(j-1)}{1-S(j-1)}\phi(j)\rho(j)
      -\phi(j+1)\rho(j+1)\mf\\
                \end{split}\end{equation}
\end{proposition}
\begin{proof}  Let $\psi(j) \deq\frac{1}{1-S(j-1)}$, then
  \begin{equation}\begin{split}
      h(j)&=
      \begin{cases}
        -\psi(j)\rho(j)\phi(j)  & \rho(j)>0 \left(\Rightarrow
          \phi(j)>0\right)\\
        \beta(2)-s(j) - \phi(j+1)\rho(j+1) &\rho(j)<0
        \left(\Rightarrow \phi(j)=0\right)
      \end{cases}
    \end{split} \mf\end{equation}
  This implies $\min\left\{0, h(j)\right\}=-\psi(j)\rho(j)\phi(j)$.
    Using this term we can see that the chosen $\lambda$ and $h(j)$ satisfies
  (\ref{bellman-ext}).
\end{proof}

\section{Conclusions}
\label{sec:conclusions}

In this paper, we have revisited the classical eviction problem,
relating it to optimal control theory and introducing the average
occupancy variant, which provides solutions and insights even for the
classical, fixed occupancy version of the problem.

A number of interesting and challenging issues remain open in the area
of eviction policies for the memory hierarchy. One objective is the
search for optimal policies (or policies with good performance
guarantees), with fixed occupancy, for general HMRM traces. In this
context, it may be worthwhile to investigate the Least Profit Rate
policy beyond the LRUSM model.

Within the LRUSM, we have considered policy design assuming a known
stack-depth distribution: what performance guarantees can be achieved
if the distribution is not known a priori, but perhaps estimated
on-line, is another intriguing question, whose answer may have
practical value for memory management in general purpose systems,
where different applications are likely to conform to different
distributions.

In this work, we have also explored forms of optimality different from
gain optimality. However, even within the LRUSM, we lack general
solutions for a finite horizon as well as for infinite horizon, if we
insist on bias optimality.

A question underlying the entire area of eviction policies remains the
choice of an appropriate stochastic model for the trace. While the
LRUSM captures temporal locality in a reasonable fashion, it
completely misses spatial locality, a property critically exploited in
hardware and software systems. Spatial locality implies that certain
subsets of the addressable items occur more frequently in short
intervals of the trace than other subsets. On the contrary, the LRUSM
is invariant under arbitrary permutations of the items. The Markov
Reference Model can capture some level of temporal and space locality,
for example if the transition graph contains regions where outward
transitions have low probability, thus corresponding to a sort of
working set. However, in real programs, the same item tends to occur
in different working sets at different times, that is, the same item
can be accessed in different states of the trace, so that the state
cannot be identified with the last item that has been accessed, as in
the MRM (see also \cite{Liberatore99} for evidence on the limitations
of Markov models). Suitable hidden Markov models do not necessarily
suffer from this limitation, which motivates further investigations of
optimal policies for the general HMRM.

\subsection*{Acknowledgments}
We wish to express our gratitude to Prof.\ Augusto Ferrante who has kindly
provided valuable expert advise on optimal control theory at many critical
junctures of this research. 

\bibliographystyle{acm}
\bibliography{lpr}
\clearpage

\appendix

\section{Bellman Equation}
\label{sec:bellman-equation}
Let $\Delta$ be a discrete dynamical system and $X$, $W$ and $Q$ denote
respectively its state, disturbance and control spaces; let $\mathcal{P}$ be the
set of all the admissible policies to control the system.
\begin{definition}
  A system is said to be of \emph{type 1} (the standard model) if the policies
  are allowed to choose the control $u$ at time $t$ only as a function of the
  state at the same time:
  \begin{equation}
    u_t=\mu\left(x_t\right) \in U\left(x_t\right) \mc
  \end{equation}
  where
  \begin{equation}
    U:\quad X \rightarrow \mathscr{P}(Q) \mc
  \end{equation}
  and
  \begin{equation}
    \mu\in \mathcal{P}:\quad X \rightarrow Q \mf
  \end{equation}
\end{definition}
\begin{definition}
  A system is said to be of \emph{type 2} (our model) if the policies are
  allowed to choose the control $u$ at time $t$ as a function of the state and
  the disturbance at the same time:
  \begin{equation}
    u_t=\mu\left(x_t,w_t\right) \in U\left(x_t,w_t\right) \mc
  \end{equation}
  where
  \begin{equation}
    U:\quad X\times W \rightarrow \mathscr{P}(Q) \mc
  \end{equation}
  and
  \begin{equation}
    \mu\in \mathcal{P}:\quad X\times W \rightarrow Q \mf
  \end{equation}
\end{definition}
\begin{definition}
  A system
  $\Delta\tone=\left(X\tone,W\tone,Q\tone,g\tone,f\tone,U\tone(\cdot,\cdot)\right)$
  of type 2 is said to be \emph{equivalent} to a system
  $\Delta\ttwo=\left(X\ttwo,W\ttwo,Q\ttwo,g\ttwo,f\ttwo,U\ttwo(\cdot)\right)$ of
  type 1 if and only if $X\tone=X\ttwo$, $W\tone=W\ttwo$, $g\tone=g\ttwo$ and
  \begin{equation}
    \begin{split}
      &\forall x \in X \quad \forall w \in W\\
      &\exists u\tone \in U\tone(x,w) : f\tone(x,u\tone,w)=y \quad \Longleftrightarrow \quad
      \exists u\ttwo \in U\ttwo(x) : f\ttwo(x,u\ttwo,w)=y  \mf
    \end{split}
  \end{equation}
\end{definition}
\begin{remark}
  Given an initial state $x_0$ and a realization of $w_t$ for two equivalent
  systems $\Delta\tone$ and $\Delta\ttwo$, and a sequence of controls $u\tone_t$
  for system $\Delta\tone$ is always possible to find $u\ttwo_t$ such that the
  state trajectories $x_t$ (and hence the costs) of the two equivalent systems
  are the same.
\end{remark}
\begin{lemma}\label{belequiv}
  For every system $\Delta\tone$ of type 2 exists a system $\Delta\ttwo$ of type
  1 s.t.\ $\Delta\tone$ and $\Delta\ttwo$ are equivalent.
\end{lemma}
\begin{proof}
  The proof is obtained by choosing as controls for $\Delta\ttwo$ the policies
  of $\Delta\tone$. Let
  \begin{align}
    X'\deq X, \qquad W'\deq W, \qquad g'&\deq g, \qquad Q\ttwo\deq
    \mathcal{P}\tone, \qquad \forall x
    \in X \quad U\ttwo(x)\deq \mathcal{P}\tone \\
    f\ttwo\left(x,u\ttwo,w\right)&\deq f\tone\left(x,u'\tone(x,w),w\right),
    \qquad \text{(since $u' \in \mathcal{P}$)} \mf
  \end{align}
  Then $\forall x \in X, \quad \forall w \in W,$
  \begin{itemize}
  \item $\forall u\tone \in U\tone(x,w)$, let $y=f\tone\left(x,u\tone,w\right)$.
    If we set $u\ttwo$ such that $u\ttwo(x,w)=u\tone$ we have
    \begin{equation}
        f\ttwo\left(x,u\ttwo,w\right)=f\tone\left(x,u\ttwo(x,w),w\right)
        =f\tone\left(x,u\tone,w\right)=y \mf
    \end{equation}
  \item $\forall u\ttwo\in U\ttwo(x)$, let $y=f\ttwo\left(x,u\ttwo,w\right)$.
    If we set $u\tone=u'(x,w)$ we have
    \begin{equation}
      f\tone\left(x,u\tone,w\right) = f\tone\left(x,u\ttwo(x,w),w\right) =
      f\ttwo\left(x,u\ttwo,w\right) = y \mf
    \end{equation}
  \end{itemize}
\end{proof}
\par{\bf Optimal cost update equations.}
Let $w_t$ be a random process with values in $W$, i.i.d.\ for different $t$'s.
Let $\Delta\tone$ be a system of type 2 and $\Delta\ttwo$ an equivalent system
of type 1.  Then the cost update equation for $\Delta\tone$ can be written as:
\begin{align}\label{beq1}
  \begin{split}
    \forall x\in X \quad J^*_\tau(x)&=\E_w\left[\min_{u\tone\in
        U\tone(x,w)}\left\{g(x,w)+J^*_{\tau-1}
        \left(f\tone\left(x,u\tone,w\right)\right)\right\}\right] \mc
  \end{split}\\
  \vec J^*_\tau &= \T\tone \vec J^*_{\tau-1} \mc
\end{align}
whereas the same equation for $\Delta\ttwo$ is:
\begin{align}\label{beq2}
  \begin{split}
    \forall x\in X \quad J^*_\tau(x)&=\min_{u\ttwo\in U\ttwo(x)}\left\{\E_w\left[
        g(x,w)+J^*_{\tau-1}\left(f\ttwo\left(x,u\ttwo,w\right)\right)\right]\right\} \mc
  \end{split}\\
  \vec J^*_\tau &= \T\ttwo \vec J^*_{\tau-1} \mf
\end{align}
\begin{remark}
  Since equivalent systems can reproduce each other's state evolution, their
  optimal costs are the same, in particular
  \begin{equation}
    \T\tone \vec J^*_{\tau-1}=\T\ttwo \vec J^*_{\tau-1} \mf
  \end{equation}
\end{remark}
We recall the classical Bellman equation theorem:
\begin{theorem}[Standard Bellman equation]\label{belorig}
  Given a dynamical system $\Delta\ttwo$ of type 1, if $\exists \lambda$ and
  $\exists \vec h$ such that
  \begin{equation}
    \lambda \vec 1 + \vec h = \T\tone \vec h \mc
  \end{equation}
  then $\lambda$ is the optimal average cost of $\Delta\tone$ and $\vec h$ are the
  differential costs of the states, i.e.
  \begin{equation}
    \forall x,y \in X \quad \lim_{\tau\rightarrow +\infty} J_\tau^*(x)-J_\tau^*(y)=h(x)-h(y) \mf
  \end{equation}
\end{theorem}
We are now ready to prove our version of the Bellman equation for a system
$\Delta\tone$ of type 2:
\begin{proof}[{\bf Proof of Prop.~\ref{bel}}]
  Consider a system $\Delta\ttwo$ equivalent to $\Delta\tone$. Applying
  Thm.~\ref{belorig} we have that, if we can solve Bellman equation for
  $\Delta\ttwo$ then we have found its optimal average cost and differential
  costs vector. Since the two systems are equivalent this implies that they are
  also the corresponding costs of $\Delta\tone$. Hence we have
  \begin{equation}
    \exists \lambda \, \exists \vec h : \lambda\vec 1 + \vec h = \T\ttwo \vec h
    \quad \Rightarrow \; \text{$\lambda$ and $\vec h$ costs for }\Delta\tone \mc
  \end{equation}
  but since $\T\ttwo \vec h=\T\tone \vec h$ we finally have
  \begin{equation}
    \exists \lambda \, \exists \vec h : \lambda\vec 1 + \vec h = \T\tone \vec h
    \quad \Rightarrow \; \text{$\lambda$ and $\vec h$ costs for }\Delta\tone \mf
  \end{equation}
\end{proof}

\end{document}